\documentclass[acmsmall]{acmart}

\AtBeginDocument{%
  }

\setcopyright{acmlicensed}
\acmJournal{PACMMOD}
\acmYear{2024} 
\acmVolume{2} 
\acmNumber{5 (PODS)} \acmArticle{209} 
\acmMonth{11}
\acmDOI{10.1145/3695827}

\newcommand{\N}{\mathbb{N}}

\newcommand{\Ber}{\text{Ber}}
\newcommand{\Geo}{\text{Geo}}

\def\a{\alpha}
\def\b{\beta}

\newtheorem{fact}{Fact}
\newtheorem{claim}{Claim}

\usepackage{titlesec}
\usepackage{enumitem}
\usepackage[linesnumbered,ruled,vlined]{algorithm2e}
\titleformat{\subparagraph}[runin]{\normalfont\normalsize\bfseries}{\thesubparagraph}{1em}{}
\titlespacing*{\subparagraph}{0pt}{\parskip}{0.5em}

\usepackage[bibliography=common]{apxproof}

\AtEndPreamble{
\newtheoremrep{claimapp}[theorem]{Claim}
\newtheoremrep{theoremapp}[theorem]{Theorem}
\newtheoremrep{factapp}[theorem]{Fact}
\newtheoremrep{lemmaapp}[theorem]{Lemma}

\renewcommand{\mainbodyrepeatedtheorem}{\emph{\textbf{[*]}}}
}
\newcommand{\poly}{\ensuremath{\text{poly}}}

\newcommand{\Pt}{\mathcal{P}}
\newcommand{\Dt}{\mathcal{D}}
\newcommand{\Qt}{\mathcal{Q}}
\begin{document}

\title{Optimal Dynamic Parameterized Subset Sampling}

\author{Junhao Gan}
\affiliation{%
  \institution{The University of Melbourne}
  \city{Melbourne}
  \country{Australia}}
\email{junhao.gan@unimelb.edu.au}

\author{Seeun William Umboh}
\affiliation{%
  \institution{The University of Melbourne}
  \city{Melbourne}
  \country{Australia}
}
\email{william.umboh@unimelb.edu.au}

\author{Hanzhi Wang}
\authornote{Work partially done at Renmin University of China and The University of Melbourne.}
\affiliation{%
 \institution{BARC, University of Copenhagen}
 \city{Copenhagen}
 \country{Denmark}}
\email{hzwang.helen@gmail.com}

\author{Anthony Wirth}
\affiliation{%
  \institution{The University of Sydney}
  \city{Sydney}
  % \state{Beijing Shi}
  \country{Australia}}
\email{anthony.wirth@sydney.edu.au}

\author{Zhuo Zhang}
\affiliation{%
  \institution{The University of Melbourne}
  \city{Melbourne}
  % \state{Texas}
  \country{Australia}}
\email{zhuo.zhang@student.unimelb.edu.au}

\renewcommand{\shortauthors}{Gan et al.}

\begin{abstract}
In this paper, we study the {\em Dynamic Parameterized Subset Sampling} (DPSS) problem in the {\em Word RAM model}. In DPSS, the input is a set,~$S$, of~$n$ items, where each item,~$x$, has a {\em non-negative integer weight},~$w(x)$. Given a pair of query parameters, $(\alpha, \beta)$, each of which is a {\em non-negative rational number}, a {\em parameterized subset sampling query} on~$S$ seeks to return a subset $T \subseteq S$ such that each item $x \in S$ is selected in~$T$, independently, with probability $p_x(\alpha, \beta) = \min \left\{\frac{w(x)}{\alpha \sum_{x\in S} w(x)+\beta}, 1 \right\}$. More specifically, the DPSS problem is defined in a dynamic setting, where the item set,~$S$, can be updated with insertions of new items or deletions of existing items. 

Our first main result is an optimal algorithm for solving the DPSS problem, which achieves~$O(n)$ pre-processing time, $O(1+\mu_S(\alpha,\beta))$ expected time for each query parameterized by $(\alpha, \beta)$, {\em given on-the-fly}, and $O(1)$ time for each update; here, $\mu_S(\alpha,\beta)$ is the expected size of the query result. At all times, the worst-case space consumption of our algorithm is linear in the current number of items in~$S$.

Our second main contribution is a hardness result for the DPSS problem when the item weights are~$O(1)$-word {\em float numbers}, rather than integers. Specifically, we reduce Integer Sorting to the deletion-only DPSS problem with float item weights. Our reduction shows that an optimal algorithm for deletion-only DPSS with float item weights (achieving all the same bounds as aforementioned) implies a $O(N)$-expected-time algorithm for sorting $N$ integers. 
The latter remains an important open problem. 
Moreover, a deletion-only DPSS algorithm which supports float item weights, with complexities worse, by at most a factor of $o(\sqrt{\log \log N})$, than the optimal counterparts, would already improve the current-best integer sorting algorithm
%by Han and Thorup
~\cite{han2002sorting}.

Last but not least, a key technical ingredient for our first main result is a set of  exact and efficient algorithms for generating Bernoulli (of certain forms) and Truncated Geometric random variates in $O(1)$ expected time with $O(n)$ worst-case space in the Word RAM model. Generating Bernoulli and geometric random variates efficiently is of great importance not only to sampling problems but also to encryption in cybersecurity. We believe that our new algorithms may be of independent interests for related research.
\end{abstract}



% \keywords{Subset Sampling, Word RAM Model, Algorithm Complexity}

% Articles V2mod208-V2mod215 use
\received{May 2024}
% \received[revised]{JuneX 2024}
\received[accepted]{August 2024}

\maketitle

\begin{toappendix}
\section{Applications Study of the DPSS Problem}
\label{sec:application}
Subset sampling is a versatile algorithmic tool that has found multifaceted applications in influence maximization~\cite{subsim_influence_guo, subsim_tods}, local clustering~\cite{wang2021approximate}, graph neural networks~\cite{wang2021approximate}, bipartite matching~\cite{bhattacharya2023graphmatching}, and computational epidemiology~\cite{germann2006mitigation}, among many others.
Several of these applications align well with the DPSS setting, where the number of sampling items and sampling probabilities change over time. Below, we name just a few. 

\subsection{Influence Maximization}
In the influence maximization, we are asked to identify $k$ seed nodes in a given social network such that the expected number of nodes influenced by these $k$ seed nodes is maximized. A commonly adopted model in influence maximization studies is the independent cascade (IC) model, where every node $u$ in the network has an independent probability $p(u,v)$ of influencing each of its out-neighbors $v$. Here, $p(u,v)$ is a numerical value in the range of $(0,1)$ determined by the properties of nodes $u$ and $v$ (e.g., the degrees of $u$ and $v$, or the edge weight of $(u,v)$). 

A set of influence maximization algorithms~\cite{subsim_influence_guo,subsim_tods,zhang2024influence} adopts the following strategy to identify seed nodes. Initially, they select a node $v$ in the network uniformly at random, terming $v$ an activated node. Then, they use subset sampling techniques at each activated node (initially, node $v$) to independently sample each in-neighbor with the corresponding sampling probability (e.g., independently sampling each in-neighbor $w$ of $v$ with probability $p(w,v)$), further terming all sampled nodes as activated nodes. This process is repeated until no more nodes can be activated. Finally, all these activated nodes are merged into a set called a reverse reachable (RR) set. Several RR sets are sampled in the network, and $k$ nodes are selected from these RR sets in a greedy manner. 

According to the aforementioned process, the query complexity of subset sampling determines the efficiency of these influence maximization algorithms. 
In particular, when we consider dynamic network structures, which are extremely common in real life, nodes and edges are continuously inserted and deleted. Each time when an edge $(s,v)$ is inserted or deleted in the graph, the degree of $v$ changes and hence, the sampling probability $p(w,v)$ for every in-neighbor $w$ of $v$ changes accordingly. In this case, DSS and weighted sampling techniques are infeasible (at least not efficient). Hence, we require efficient DPSS techniques to sample the in-neighbors of activated nodes independently. 

\subsection{Local Clustering}
In the local clustering scenario, given an undirected graph $G=(V,E)$ with a seed node $s\in V$, we aim to identify a high-quality cluster around node $s$. Here, a cluster refers to a group of nodes that are closely connected to each other. A line of studies~\cite{spielman2004Nibble,andersen2006FOCS,yang2019TEA, wang2021approximate} adopts a three-phase approach. They first compute random-walk probabilities $\pi(s,u)$ from the given seed node $s$ to all nodes $u$ in the graph. Next, they sort all nodes $u$ in the graph in descending order of $\pi(s,u)/d_u$, where $d_u$ denotes the degree of node $u$. Finally, they inspect the prefix set of all nodes in the sorted order, compare the cluster quality of each prefix set, and return the cluster with the highest quality. As a consequence, the computational efficiency of random-walk probabilities in the first step determines the query efficiency of local clustering tasks. 

Wang et al.~\cite{wang2021approximate} adopt subset sampling techniques in the computation of random-walk probabilities to speed up the computational efficiency. Specifically, they design a novel ``push'' strategy, initially placing one probability mass at the seed node $s$, then pushing probability mass from $s$ along out-edges to all nodes in the graph step by step. In the push operation at any node $u$, Wang et al. adopt a subset sampling technique to independently sample every out-neighbor $v$ of $u$ with probability $\frac{\mathrm{A}_{uv}}{d_{\mathrm{out}}(u)}$, where $\mathrm{A}_{uv}$ denotes the edge weight of edge $(u,v)$ and $d_{\mathrm{out}}(u)$ denotes the out-degree of node $u$. Notably, for any node $u\in V$, the sampling probability $\frac{\mathrm{A}_{uv}}{d_{\mathrm{out}}(u)}$ of every out-neighbor $v$ of $u$ changes together when $d_{\mathrm{out}}(u)$ is updated (e.g., when an edge adjacent to u is inserted or deleted in the graph). 

While Wang et al. only apply their local clustering algorithm to static unweighted graphs, most real-world networks have dynamic weighted structures, where the weight of each edge can indicate any strength measures of the relationships between two nodes. Additionally, a Science'16 paper~\cite{benson2016Science} and its follow-up~\cite{yin2017MAPPR} point out that conducting local clustering on motif-based weighted graphs can enhance the quality of detected clusters, where a motif refers to a kind of graph pattern (e.g., a triangle). Consequently, in these real-world scenarios, we require DPSS algorithms to identify high-quality local clusters efficiently.

\end{toappendix}
\section{Introduction}

Sampling is arguably one of the most important techniques for solving problems on massive data.
It can effectively reduce 
the problem size 
while still providing reliable estimation for certain metrics 
and largely retaining the statistical properties of the original data.
As such, 
sampling has been widely adopted for various problems in different research fields, including 
anomaly detection~\cite{zenati2018adversarially, juba2015principled, mai2006sampled}, feature selection~\cite{kumar2014feature,vu2019feature}, sketching algorithms~\cite{rusu2009sketching}, network measurement~\cite{haveliwala1999efficient,xing2004weighted,wang2021approximate}, distributed computing~\cite{fleming2018stochastic,esfandiari2021almost}, cybersecurity~\cite{baldwin2017contagion,kent2016cyber}, machine learning~\cite{liberty2016stratified,hajar2019discrete}, and computational geometry~\cite{de2000computational,haussler1986epsilon}.
In general, sampling methods fall into two categories: {\em Weighted Sampling} and {\em Subset Sampling}. 
In Weighted Sampling, the input is a set of~$n$ items,~$S$, and each item,~$x$, has a non-negative weight,~$w(x)$. 
It aims to sample {\em just one} item such that each item~$x$ is chosen with probability $\frac{w(x)}{\sum_{x\in S} w(x)}$. 
On the other hand, in the vanilla setting of Subset Sampling, the input is a set $S$ of $n$ items, each with a {\em probability}, $p(x)$, the goal is to sample a {\em subset} $ T \subseteq S$ such that each item $x$ has probability $p(x)$ to be selected in $T$, independently.

A natural extension of Subset Sampling is to consider the {\em dynamic}
setting, where the item set~$S$ can be updated by either insertions of new items or deletions of existing ones.
This is called the {\em Dynamic Subset Sampling} (DSS) problem.
Yi et al.~\cite{yi2023optimal} proposed an {\em optimal} algorithm, called ODSS, for solving the DSS problem in a special Real RAM model~\footnote{Although Yi et al.~\cite{yi2023optimal} claimed that their algorithm can work on the Word RAM model, the technical details of the generation of Truncated Geometric random variates have been overlooked.} 
which allows both logarithm and rounding operations can be performed in $O(1)$ time.
Specifically, their ODSS algorithm achieves~$O(n)$ preprocessing time, $O(1 + \mu)$ expected query time, where~$\mu$ is the expected size of the sampling result, and supports each update in~$O(1)$ time. 

\vspace{1mm}
\noindent
{\bf Dynamic Parameterized Subset Sampling.}
In this paper, we study a more general setting of the Subset Sampling
problem, called {\em Dynamic Parameterized Subset Sampling} (DPSS), in the Word RAM model. 
The DPSS problem is defined as follows.
Consider a set $S$ of $n$ items, which is subject to updates by either item insertions or deletions;
each item $x\in S$ has a {\em non-negative integer} weight $w(x)$.
Each {\em Parameterized Subset Sampling} (PSS) query is associated with a pair of non-negative rational parameters $(\a, \b)$, and asks for a subset $T \subseteq S$ such that 
each item $x\in S$ is selected in $T$, independently, with probability $ p_x(\a, \b) = \min\left\{ \frac{w(x)}{\a \sum_{x\in S} w(x) + \b}, 1 \right\}$. 

\vspace{1mm}
\noindent
\underline{\em Distinctions between DPSS and DSS.} 
Comparing to the aforementioned
DSS problem, there are several crucial distinctions from DPSS.
First, in DSS, the sampling probability of each item is known in advance, 
while in DPSS, the sampling probabilities of the items 
are parameterized by the query parameters $(\a, \b)$  
{\em on the fly} during query time.
Second, 
in DSS, the sampling probabilities of the items are 
{\em decoupled}, in the sense that updates on the sampling probability of items would not affect each other. 
However, in contrast, in DPSS, 
the sampling probabilities of the items are {\em inter-correlated} such that an update to an item's weight would affect the sampling probabilities of all other items.
In other words, the impact of each update to the item set $S$ in DPSS 
is a lot larger.
These distinctions, therefore, make DPSS more challenging than the  DSS problem. 
To see this, the existing optimal ODSS algorithm requires $\Omega(n)$ time to support an update in the DPSS setup even just with fixed and known query parameters $(\a, \b)$ at all times. 

\vspace{1mm}
\noindent
\underline{\em Applications of DPSS.}
In Appendix~\ref{sec:application}, we provide two detailed case studies on applications of DPSS in Influence Maximization~\cite{subsim_influence_guo, subsim_tods} and Local Clustering~\cite{wang2021approximate}.
These two applications further illustrate that current algorithms (such as DSS algorithms) cannot address these cases since the sampling probabilities of items may change simultaneously after just one update.

\vspace{1mm}
\noindent
{\bf Main Result 1: An Optimal DPSS Algorithm.}
Our first contribution is an optimal algorithm for solving the DPSS problem in the Word RAM model.
The main result is summarized as follows:
\begin{theorem}
\label{thm:halt}
Consider a set $S$ of $n$ items;
there exists a data structure for solving the Dynamic Parameterized Subset Sampling (DPSS) problem, which achieves:
\begin{itemize}[leftmargin = *]
\item \textsc{Pre-processing time}: $O(n)$ worst-case time to construct the data structure on $S$;
\item \textsc{Query time}: 
$O(1+\mu_S(\a,\b))$ expected time to return a sample for each Parameterized Subset Sampling (PSS) query on $S$ with parameters $(\a, \b)$, where $\mu_S(\a, \b)$ is the expected size of the sample;
\item \textsc{Update time}: $O(1)$ worst-case time for each item insertion or deletion. 
\end{itemize}
At all times, the worst-case space consumption of such a data structure is bounded by $O(n)$, where $n$ denotes the current cardinality of~$S$. 

\end{theorem}

\noindent
\underline{\em An Overview of Our Algorithm.}
Given a pair of Parameterized Subset Sampling (PSS) query parameters, $(\alpha, \beta)$, the crucial idea of our algorithm is to convert the 
PSS query on~$S$ into a~$O(1)$ number of PSS query instances (with possibly different query parameters on possibly different item sets), 
of which the PSS results can be converted back into a feasible PSS result on~$S$ by {\em rejection sampling}.
And each of these~$O(1)$ query instances can be solved optimally in $O(1 + \mu)$ expected time, where $\mu$ is its expected sample size. 
To achieve such a conversion, and to solve those converted PSS query instances optimally,
we propose a novel data structure, called {\em Hierarchy + Adapter + Lookup Table} (HALT).
As its name suggests, it consists of three main components:
(i) a three-level sampling hierarchy, 
(ii) a set of dynamic adapters, and
(iii) a static lookup table.

Specifically, the sampling hierarchy converts the PSS query on~$S$ into~$O(1)$ PSS query instances on other item sets, 
while the lookup table can answer PSS queries on a special family of item sets.
Last, but not  least, as the lookup table is static, while the PSS queries are parameterized on the fly, our dynamic adapters play a crucial role to 
bridge the hierarchy and the lookup table.
Specifically,
the adapters, dynamically translate the PSS query instances, during query time, into a form for which  our lookup table is applicable.

\vspace{1mm}
\noindent
{\bf Main Result 2: A Hardness Result on DPSS with Float Item Weights.} 
While our 
HALT algorithm is optimal for DPSS 
with integer item weights,
interestingly, if the item weights are floating-point numbers, each represented by a $O(1)$-word exponent and a $O(1)$-word mantissa, 
the DPSS problem suddenly becomes a lot  more difficult.
In fact, this difficulty arises even in the {\em deletion-only} setting where the updates are deletions only.
We demonstrate this by showing that DPSS can be used to obtain an algorithm for Integer Sorting:

\begin{theorem}\label{thm:hardness}
Suppose there exists an algorithm for deletion-only DPSS on a set of~$N$ items with float weights that has pre-processing time $t_p(N)$, expected query time $(1 + \mu) \cdot t_q(N)$, and deletion time~$t_{del}(N)$. 
Then the Integer Sorting problem for~$N$ integers can be solved in $t_p(N) + O(N \cdot (t_q(N) + t_{del}(N)))$ expected time on the $d$-bit Word RAM model with $d \in \Omega(\log N)$.
\end{theorem}

In particular, this implies that if there is an optimal algorithm for deletion-only DPSS---i.e.,~one with $t_p(N) \in O(N)$, and $t_q(N), t_{del}(N) \in O(1)$---then 
sorting  $N$ integers (each represented with $O(1)$ words) can be achieved in 
%there is an optimal algorithm for integer sorting---i.e.,~one with 
$O(N)$ expected running time on the $d$-bit Word RAM model with $d \in \Omega(\log N)$.
However, solving the Integer Sorting problem in $O(N)$ expected time for every word-length $d \in \Omega(\log N)$ remains an open problem~\cite{belazzougui2014expected}.
Moreover, an algorithm for deletion-only DPSS (with float item weights) with $t_p(N) \in o(N \cdot \sqrt{\log\log N})$, and $t_q(N), t_{del}(N) \in o(\sqrt{\log\log N})$ already improves the current-best $O(N \cdot \sqrt{\log \log N})$ expected running time bound of Han and Thorup~\cite{han2002sorting}. 

\vspace{1mm}
\noindent
{\bf Main Result 3: An Efficient Algorithm for Truncated Geometric Variate Generation.}
Our third contribution is an efficient algorithm for generating {\em Truncated Geometric} random variates in the Word RAM model. 
Specifically, for $p \in (0, 1)$ and a positive integer $n$, 
a truncated geometric distribution, $\text{T-Geo}(p, n)$,  returns a value $i \in \{1, \ldots, n\}$ with probability $\frac{p (1-p)^{i-1}}{1 - (1 - p)^n}$. 
As can be seen from the literature, c.f., the aforementioned ODSS algorithm and from our HALT algorithm,
generating truncated geometric random variates is of great importance and a crucial building block in many sampling-related algorithms.
Somewhat surprisingly, given its importance, we are not aware of any previous work on generating truncated geometric variates in $O(1)$ expected time with $O(n)$ worst-case space in the Word RAM model. 
We give the following theorem to make it documented:

\begin{theorem}\label{thm:trun-geo}
Let $p\in (0,1)$ be a rational number whose numerator and denominator fit in~$O(1)$ words, and~$n$ be a positive integer fitting in~$O(1)$ words. 
There exists an algorithm that generates a random variate from $\text{\em T-Geo}(p,n)$ in~$O(1)$ expected time and~$O(n)$ space in the worst case in the Word RAM model.
\end{theorem}

Considering the importance of efficient geometric random variate generation to many research fields, we believe that  
%As mentioned earlier, generating geometric random variates efficiently is of
%great importance to many research fields. Hence, we think 
our new algorithm is of independent interest.

%We have deferred 
The proofs of some statements (indicated by~\emph{\mainbodyrepeatedtheorem}) have been deffered to the appendix. 
Meanwhile, the pseudocodes of our DPSS algorithms can be found in the appendix.

\section{Preliminaries}

In this section, we first introduce the {\em Word RAM model}~\cite{fredman1993surpassing} and then formally define the problem of {\em Dynamic Parameterized Subset Sampling} (DPSS).

\subsection{The Word RAM Model} 

In the Word RAM model~\cite{fredman1993surpassing}, each {\em word} consists of $d$ bits.
Each memory address is represented by~$O(1)$ words and hence, there are at most $2^{O(d)}$ memory addresses, each corresponding to a unique memory cell.
Each memory cell stores a word of content which can be considered as an integer in the range of $\{0, 1, \ldots, 2^d-1\}$.
The content of each memory cell can be accessed by its memory address 
in $O(1)$ time. 
Furthermore, 
given two one-word integers $a$ and $b$,
each of the following basic arithmetic operations on $a$ and $b$ can be performed in $O(1)$ time:
addition ($a+b$), subtraction ($a - b$), multiplication ($a \cdot b$),
division with rounding ($\lfloor a / b \rfloor$), comparison ($> , < , =$), truncation, and bit operations (e.g., finding the index of the highest or lowest non-zero bit in $a$).
Moreover, every {\em long integer} is represented by an array of words, and every {\em float number} is represented by its exponent and mantissa, where each takes $O(1)$ words. 
Finally, as a convention, we assume that 
a uniformly random word of $d$ bits can be generated in $O(1)$ time.
%generating a random integer containing $d$ random bits takes $O(1)$ time.
%
Throughout this paper, the space consumption of an algorithm is measured in words. 

The following fact about the Word RAM model is useful:
\begin{factapprep}\label{fact:sorted}
In the $d$-bit Word RAM model, consider 
the integer universe $U = \{0, 1, \ldots, d-1\}$, and 
a set $I \subseteq U$ of $n$ integers;
%from the universe $U = \{0, 1, \ldots, d-1\}$;
there exists a data structure that maintains all the integers in $I$ in a {sorted} linked list and 
supports each of the following operations in $O(1)$ worst-case time:
\begin{itemize}
\item each update (either an insertion of a new integer or a deletion of an existing integer) to $I$;
\item finding the predecessor or successor for any integer $q \in U$ in $I$. 
\end{itemize}
Such a data structure consumes~$O(n)$ space at all times, where~$n$ is the current cardinality of~$I$.
\end{factapprep} 
\begin{proof}
We provide an implementation to proof Fact~\ref{fact:sorted}. 

\noindent{\bf The Data Structure.} The input is a set $I$ of $n$ integers from the universe $U = \{0, 1, \ldots, d-1\}$. We can construct the following data structure, which includes the following four parts: a bitmap $\mathcal{M}$, a pointer array $\mathcal{P}$, a menu array $\mathcal{Q}$ and a sorted linked list $\mathcal{L}$.
\begin{itemize}[leftmargin = *]
\item \underline{\em The bitmap $\mathcal{M}$.} The bitmap $\mathcal{M}$ is a 1 word integer which contains $d$ bits and the $i$-th bit indicates whether the integer $i$ is in $I$.
\item \underline{The sorted linked list $\mathcal{L}$.} A sorted list that contains the integers in $I$.
\item \underline{The pointer array $\mathcal{P}$.} $\mathcal{P}$ is a dynamic array with $n$ elements and each element stores the \textit{pointer} to an integer in $\mathcal{L}$. 
\item \underline{The menu array $\mathcal{Q}$.} $\mathcal{Q}$ contains $d$ integers and each of the integers is in the range of $[1,n]$. For each integer $i\in I$, the $i$-th element of $\mathcal{Q}$ is denoted as $\mathcal{Q}[i]$. We make sure the $\mathcal{Q}[i]$-th element of $\mathcal{P}$, i.e. $\mathcal{P}[\mathcal{Q}[i]]$, stores the pointer to $i$.
\end{itemize}

\noindent{\bf The Construction of the Data Structure.}
The above data structure can be constructed in $O(n)$ worst case time and in $O(n)$ space with the following four steps.
\begin{itemize}[leftmargin = *]
\item \underline{\em Step 1.} Construct the bitmap $\mathcal{M}$.
\item \underline{\em Step 2.} Initialize $\mathcal{Q}$ with all 0 elements. 
Observe that each integer in $\mathcal{Q}$ occupies at most $O(\log n)$ bits of storage and in total $\mathcal{Q}$ takes $O(\log n)$ words and can be initialized in $O(\log n )$ time.
\item \underline{\em Step 3.} For each integer, $i$ in $I$ generate a node of the sorted linked list and push back its pointer to $\mathcal{P}$. At the same time, \textit{update} the menu array $\mathcal{Q}$ by setting $\mathcal{Q}[i]$ as the index of $i$-th pointer in $\mathcal{P}$.
\item \underline{\em Step 4.} Link the nodes in order to construct $\mathcal{L}$. Given an integer $q\in U$, we find its successor, denoted as $q'$, in $I$ with the following steps:
\begin{itemize}[leftmargin = *]
\item $u\leftarrow \mathcal{M}>> q$, which shifts all the lower $q$ bits out;
\item If $u=0$, report that the successor of $q$ does not exist in $I$;
\item Otherwise, compute $q'=q+\log_2(u\mathbin \&\neg (u-1))$;
\item $q'$ is the successor of $q$ in $I$. Thus $\mathcal{P}[\mathcal{Q}[q']]$ stores the pointer of its successor node.
\end{itemize}
\end{itemize}

Clearly, the construction algorithm above only takes $O(n)$ worst case running time and takes $O(n)$ space.

\noindent{\bf Maintain the Data Structure with Updates.}
Now we consider how to maintain this data structure when an insertion of a new integer or a deletion of an existing integer happens. We first suppose the insertion and deletion does not change the value of $\lceil \log_2 n\rceil$, thus we do not need to rebuild $\mathcal{Q}$ and $\mathcal{P}$.

\noindent\underline{\em Insertion.}
When a new integer $q$ is inserted to $I$, the corresponding bit in $\mathcal{M}$ is set as 1. Then build a linked sorted list node of $q$, tail insert its pointer to $\mathcal{P}$ and set the the $q$-th element of $\mathcal{Q}$ as the index. Then we find the pointer of $q$'s predecessor or successor and insert the node of $q$ to the correct location in the linked list to maintain it to be sorted.

\noindent\underline{\em Deletion.}
To delete an existing integer $p$, the pointer of its corresponding linked list node can be found with $\mathcal{P}[\mathcal{Q}[p]]$ and it can be removed from the sorted linked list. Then the bitmap $\mathcal{M}$  can be maintained by setting the corresponding bit as $0$. To maintain $\mathcal{P}$ and $\mathcal{Q}$, locate the index of the associated element in $\mathcal{P}$, swap it with the element at the end, and then pop out the last pointer from $\mathcal{P}$. When we do the swapping, the pointer that was swapped also needs to have its reference in $\mathcal{Q}$ updated. All of the above operations can be done in $O(1)$ time.

When the size of $I$ doubles or halves, we can rebuild the entire data structure. The rebuilding cost can be charged to those updates, making each update cost $O(1)$ amortized. Here, we can use de-amortization techniques to ensure the complexity becomes $O(1)$ in the worst case.

\end{proof}

\vspace{-1mm}
\subsection{Dynamic Parameterized Subset Sampling}

Consider a set of $n$ {\em items} $S = \{x_1, \ldots, x_n\}$; each item has a {\em weight} $w(x_i)$
which is a {\em non-negative} integer. 
Without loss of generality, 
we assume that the {largest} possible value of $n$, denoted by $n_{\max}$
 and the 
largest possible value of the item weights, denoted by $w_{\max}$
can be represented with one word (of $d$ bits) in the Word RAM model. 
Therefore, $d \in \Omega(\log (n_{\max} \cdot w_{\max}))$ holds all the time.

Define $W_S(\alpha, \beta) = \alpha \cdot \sum_{x \in S} w(x) + \beta$ as the {\em parameterized total weight function} of all the items in~$S$ by two parameters~$\alpha$ and~$\beta$, where both~$\alpha$ and~$\beta$ are {\em non-negative rational numbers} and each can be represented by a $O(1)$-word numerator and a~$O(1)$-word denominator.
Thus, the parameterized total weight $W_S(\alpha, \beta)$ of~$S$ fits in~$O(1)$ words. 
Furthermore, if the value of $\sum_{x \in S} w(x)$ is pre-computed and maintained, the value of $W_S(\alpha, \beta)$ can be computed in~$O(1)$ time.

Given 
a pair of parameters $(\alpha, \beta)$,
a {\em parameterized subset sampling query} on $S$ 
asks for
a subset $T \subseteq S$ such that each item $x \in S$ is selected in $T$ independently, with probability  
$p_x(\alpha, \beta) = \min \left\{\frac{w(x)}{W_S(\alpha, \beta)}, 1 \right\}$.
Moreover, the expected size of $T$ is defined as $\mu_S(\alpha, \beta) = \sum_{x \in S} p_x(\alpha, \beta)$.
Observe that 
since the parameters $(\alpha, \beta)$ are given on-the-fly in a parameterized subset sampling query, the sampling probability of each item can be vastly different with different parameters $(\alpha, \beta)$. 

In this paper, we study the problem of parameterized subset sampling on a {\em dynamic setting}, where the item set $S$ can be updated by insertions of new items (associated with a weight) and deletions of existing items from $S$.
We call this problem the {\em Dynamic Parameterized Subset Sampling} (DPSS) problem.
It is worth mentioning that in the DPSS problem, the sampling probability of each item is not only affected by the parameters $(\alpha, \beta)$, but also by the updates on $S$.
To see this, consider a fixed pair $(\alpha, \beta)$; before and after the insertion of an item with a very large weight, the value of $W_S(\alpha, \beta)$ can change significantly and hence, such an update dramatically changes the sampling probabilities of all the items. 

Our goal is to design a data structure for solving the DPSS problem such that: (i) it can be constructed in $O(n)$ time and consumes $O(|S|)$ space at all times, (ii) it can answer every parameterized subset sample query in $O(1 + \mu_S(\alpha, \beta))$ expected time, and 
(iii) it supports each update to $S$ in $O(1)$ worst-case time.

\section{Sampling Random Variates in the Word RAM Model}
\label{sec:random}

In this section, we discuss how to exactly and efficiently sample, in the Word RAM model,
the key random variates in our optimal DPSS algorithm.
Due to the importance of random variate generation, we believe that these new algorithms would be of independent interest to other problems.

As we see in the next section, 
our algorithm needs to generate a number of variants of {\em Bernoulli} and {\em Geometric} random variates. We explore each in turn.
\subsection{Bernoulli Random Variate Generation in the Word RAM Model}
\label{sec:ber}

Consider a real number $p\in(0,1)$, 
the {\em Bernoulli distribution} parameterized by $p$, denoted by $\Ber(p)$, takes values in $\{0,1\}$ with probability $\Pr[\Ber(p)=1]=1-\Pr[\Ber(p)=0]=p$. 
In particular, our algorithm needs to generate the following three types of Bernoulli random variates, $\text{Ber}(p)$: 
\begin{description}[leftmargin = *]
\item[type (i)] $p$ is a {\em rational number} with numerator and denominator both fitting in~$O(1)$ words;
\item[type (ii)] $p = p^*$, where $p^* = \frac{1 - (1 - q)^n}{n \cdot q}$, where $q$ is a rational number represented by a pair of $O(1)$-word numerator and denominator, $n$ is a $O(1)$-word integer, and $n\cdot q \leq 1$;
\item[type (iii)] $p = \frac{1}{2p^*}$, with~$p^*$ as in type~(ii). 
\end{description}
For type (i), Bringmann and Friedrich~\cite{bringmann2013exact} proved the following result:
\begin{fact}[\cite{bringmann2013exact}]
\label{lem:ber2}
If~$p$ is a rational number with numerator and denominator being $O(1)$-word integers, then $\Ber(p)$ can be evaluated in  $O(1)$ expected running time, and $O(1)$ worst-case space.
\end{fact}
To the best of our knowledge,  there is no previous work showing how to generate a realization of~$\Ber(p)$ of either  type~(ii) or type~(iii) in~$O(1)$ expected time with~$O(n)$ worst-case space.  
However, as we see in Section~\ref{sec:main-algo}, these two types of Bernoulli random variates play a crucial role in our optimal DPSS algorithm.
We prove the following theorem:

\begin{theorem}\label{thm:ber}
Consider a Bernoulli random variate,~$\Ber(p)$, in the form of either $p = p^*$ or $p = \frac{1}{2p^*}$, with  
$p^* = \frac{1 - (1 - q)^n}{n \cdot q}$, where $q$ is a rational number represented by a pair of $O(1)$-word numerator and denominator, $n$ is a $O(1)$-word integer, and $n\cdot q \leq 1$;
$\Ber(p)$ can be generated in $O(1)$ expected time with $O(n)$ worst-case space.
\end{theorem}

\begin{definition}[$i$-Bit Approximation]
Consider a value~$p$, possibly comprising an infinite sequence of bits. An {\em $i$-bit approximation} of~$p$, denoted by~$\tilde{p_i}$, is an approximation value of the~$i$ most-significant bits in~$p$,~i.e., satisfies $|\tilde{p_i} - p| \leq 2^{-i}$, while~$\tilde{p_i}$ itself consists of at most~$i + 1$ bits.
\end{definition}
We prove Theorem~\ref{thm:ber} 
via the following algorithmic framework of Bringmann and Friedrich~\cite{bringmann2013exact}: 
\begin{fact}
[\cite{bringmann2013exact,flajolet1986complexity}]
\label{lem:ber1}
Consider a Bernoulli random variate,~$\Ber(p)$ with rational parameter $p\in(0,1)$. Suppose the numerator and the denominator of $p$ both can be represented with $O(n)$ words, and $p$ can be evaluated in $O(\poly(n))$ time and the $i$-bit approximation,~$\tilde{p_i}$, of~$p$ can be evaluated in~$O(\poly(i))$ time. 
Then~$\Ber(p)$ can be generated in~$O(1)$ expected time with~$O(n)$ worst-case space in the Word RAM model. 
\end{fact}

\begin{lemmaapprep}
\label{lem:berppp1}
These two claims hold simultaneously:
(i) the numerator and the denominator of $p=p^*$ both can be represented with $O(n)$ words and $p$ can be evaluated in $O(poly(n))$ time; and
(ii) there exists an algorithm that can evaluate 
each $i$-bit approximation $\tilde{p_i}$ of $p = p^*$, satisfying all conditions stated in Theorem~\ref{thm:ber},  
in $O(\poly(i))$ time. 
\end{lemmaapprep}

\begin{proof}
We first prove claim (i). 
Let $A$ and $B$ denote the numerator and denominator of $q$. By definition, both $A$ and $B$ fit in $O(1)$ words.
Then it can be verified that the numerator and the denominator of $p^*$ can be computed as $B^n-
(B-A)^n$ and $B^n-nAB^{n-1}$ in $O(\poly(n))$ time, respectively. 
And each of them fits in $O(n)$ words as they are just multiplication results of $n$ constant-word integers.  
As a result, $p^*=\frac{1-(1-q)^n}{n\cdot q}$ is a rational number with $O(n)$-word numerator and denominator and it can be evaluated in $O(\poly(n))$ time.

Next we prove claim (ii). According to the binomial theorem, $1-(1-q)^n=\sum_{j=1}^{n}\binom{n}{j}(-1)^{j+1} q^j$. Thus, $p^*=\frac{1-(1-q)^n}{n\cdot q}$ can be written as $p^*=\sum_{j=1}^{n} a_j$ , where $ a_j=\frac{(-1)^{j+1}q^{j-1}(n-1)!}{j!(n-j)!}$. Note that $\frac{(n-1)!}{(n-j)!}\leq n^{j-1}$ and $nq\leq 1$. We can bound the absolute value of $a_j$ with $|a_j|\leq \frac{1}{j!}\leq \frac{1}{2^{j-1}}$. Further, $|\sum_{j=i+3}^n a_j|\leq \sum_{j=i+3}^n |a_j|\leq \frac{1}{2^{i+1}}$. In other words, the sum of $a_j$ for $1\leq j\leq i+2$ achieves a good approximation of $p^*$. Denoting $\bar{p_i}=\sum_{j=1}^{i+2} a_j$, we have $|p^*-\bar{p_i}|\leq \frac{1}{2^{i+1}}$.

Now we discuss how to compute $\bar{p_i}$ in the Word RAM model. We perform the computation with {\em working precision} $r$, which means we work with floating point numbers and each of these floating point numbers contains an exact exponent encoded by an array of words and a mantissa which is an array of $\lceil r/d \rceil$ words. 
Each mathematical operation between two floating point numbers under working precision $r$ can be done in $O(poly(r))$ time. 
Since $\bar{p_i}$ can be evaluated with $O(poly(i))$ such floating point mathematical operations, 
thus an approximation (with working precision $r$) of $\bar{p_i}$, denoted by $\bar{p_i}(r)$, can be computed in $O(poly(i,r))$ time. 
%Denote the calculation result of $\bar{p_i}$ in the Word RAM model under working precision $r$ as $\bar{p_i}(r)$. 
Analogous to the analysis in Appendix C of ~\cite{bringmann2013exact}, it suffices to choose $r=O(poly(i))$ to achieve an absolute error  $|\bar{p_i}(r)-\bar{p_i}|\leq \frac{1}{2^{i+1}}$. 

As a result, 
%With the discussion above, we can output $\bar{p_i}(r)$ as the $i$-bit approximation $\tilde{p_i}$ of $p^*$ because 
$|\bar{p_i}(r)-p^*|\leq |\bar{p_i}(r)-\bar{p_i}|+|p^*-\bar{p_i}|\leq \frac{1}{2^{i+1}} + \frac{1}{2^{i+1}} =  \frac{1}{2^{i}}$. 
Therefore, $\tilde{p_i} = \bar{p_i}(r)$ is a valid $i$-bit approximation of $p^*$, and the overall computational cost is bounded by $O(\poly(i))$.
%The calculations can be done in $O(poly(i))$ running time.
\end{proof}

\begin{lemmaapprep}
\label{lem:berppp2}
These two claims hold simultaneously:
(i) the numerator and the denominator of $p=\frac{1}{2p^*}$ both can be represented with $O(n)$ words and $p$ can be evaluated in $O(poly(n))$ time; and 
(ii) there exists an algorithm that can evaluate 
each $i$-bit approximation $\tilde{p_i}$ of $ p = \frac{1}{2p^*}$, satisfying all conditions stated in Theorem~\ref{thm:ber},  
in $O(\poly(i))$ time.
\end{lemmaapprep}

\begin{proof}
Analogous to the proof of Lemma~\ref{lem:berppp1}, the numerator and the denominator of $p=\frac{1}{2p^*}$ fit in $O(n)$ words and $p$ can be evaluated in $poly(n)$ time. Thus, claim (i) holds.

To prove claim (ii), 
by Lemma~\ref{lem:berppp1}, for an arbitrary integer $r \geq 3$, an $r$-bit approximation $\tilde{p^*_{r}}$ of $p^*$ can be computed in $O(poly(r))$ time, such that  
%where $\tilde{p^*_{r}}$ is a $r$-bit approximation of $p^*$, 
%,i.e. 
$|\tilde{p^*_{r}}-p^*|\leq \frac{1}{2^{r}}$.
%Then we can run 
By Newton's method~\cite{brent2010modern}, 
%to calculate 
an approximate reciprocal of $\tilde{2p^*_{r}}$, denoted by $u(r)$, can be computed in $O(poly(i))$ time with error at most $\frac{1}{2^{i+1}}$, 
%Denote the approximate reciprocal of $\tilde{2p^*_{r}}$ as $u$; 
%Then we have 
namely, $|u(r)-\frac{1}{\tilde{2p^*_{r}}}|\leq \frac{1}{2^{i+1}}$. 

Therefore, the absolute error, $|p - u(r)|$, is bounded by
%If we output $u$ as the result, the error is at most 
$ 
|\frac{1}{2p^*}-u(r)|
\leq |\frac{1}{2p^*}-\frac{1}{2\tilde{p^*_{r}}}|+|u(r)-\frac{1}{2\tilde{p^*_{r}}}|
\leq \frac{|2p^*-2\tilde{p^*_{r}}|}{2p^*\cdot2\tilde{p^*_{r}}}+\frac{1}{2^{i+1}}
\leq \frac{1}{2^{r-1}}+\frac{1}{2^{i+1}}
$, where
the last inequality comes from the upper bound of $p^*\geq 1-1/e\geq \frac{1}{2}$ and $\tilde{p^*_{r}}\geq p^*-\frac{1}{2^r}\geq \frac{1}{2}$.

Thus, it suffices to choose $r = i+2$ to obtain an $i$-bit approximation $u(r)$ of $p$, and the overall computational cost is bounded by $O(poly(i))$.
\end{proof}

\subsection{Geometric Random Variate Generation in the Word RAM Model}
\label{sec:geo}

Consider a parameter~$p\in(0,1)$. 
The {\em Geometric distribution} $\Geo(p)$ takes {\em discrete} values in $\N$.
Specifically, $\Pr[\Geo(p)=i]=p(1-p)^{i-1}$, for $i\in \N$. 
As the value of a geometric variate can be arbitrarily large (though it is not really likely),
%the probability of this case happening is very small), 
thus in theory, it is impossible to generate variates from a geometric distribution with worst-case bounded space. 

\vspace{2mm}
\noindent{\bf Bounded Geometric Variates.}
Nonetheless, in many application scenarios, including our optimal DPSS algorithm, 
geometric random variates within a bounded range are sufficient.
This is known as the \textit{Bounded Geometric distribution}.
Given an integer $n\in \N$, the bounded geometric distribution, denoted by
$\text{B-Geo}(p,n)$, is defined as $\min\{n,\Geo(p)\}$ with distribution: 
\begin{equation*}
\Pr[\text{B-Geo}(p,n)=i]=
\begin{cases}
p(1-p)^{i-1}& i\in\{1,\cdots,n-1\}\,; \\
(1-p)^{n-1}& i=n\,. 
\end{cases}
\end{equation*}

Bringmann and Friedrich~\cite{bringmann2013exact} studied the generation of bounded geometric random variates in the Word RAM model and gave the following theorem:
\vspace{-1mm}
\begin{fact}[\cite{bringmann2013exact}]
\label{lmm:bounded-geo}
Let $p\in (0,1)$ be a rational number with numerator and denominator fitting in~$O(1)$ words and~$n$ be a positive integer fitting in~$O(1)$ words. 
A random variate from $\text{\em B-Geo}(p,n)=\min\{n,\text{\em Geo}(p)\}$ can be realized in~$O(1)$ expected time and~$O(n)$ space in the worst case.
\end{fact}

\noindent
{\bf Truncated Geometric Variates.}
Besides the Bounded Geometric variates,
the so-called {\em Truncated Geometric} random variates are also of great importance to many applications. 
As we see in Section~\ref{sec:main-algo}, it is a crucial building block in our optimal DPSS algorithm.
Given~$p \in (0,1)$ and $n \in \N$, the truncated geometric distribution, denoted by $\text{T-Geo}(p, n)$, defined on $\{1, \ldots, n\}$, has distribution:
$$\text{Pr}[\text{T-Geo}(p,n) = i] = \frac{p(1-p)^{i-1}}{1 - (1 - p)^n}\,, \text{~for } i \in \{1,\ldots, n\}\,.$$

The probability function of truncated geometric distribution can be interpreted as a conditional probability as follows.
Consider sampling each item from a bucket of $n$ items with probability $p$ independently. Then $\text{T-Geo}(p,n)$ is identically distributed as the smallest index of the sampled items, conditioned on at least one item being sampled.

\vspace{2mm}
\noindent
{\bf Proof of Theorem~\ref{thm:trun-geo}.}
We consider the following possible cases:
\begin{description}[leftmargin = *]
\item[Case~1 ($n \leq 2$):]  
This is the trivial case to generate $i\sim \text{T-Geo}(p, n)$. 
If $n=1$, then $\text{T-Geo}(p, n) = 1$ always holds.
If $n = 2$,
it is can be verified that $\text{T-Geo}(p,n)$ returns either $i = 1$ with probability $\frac{1}{2 - p}$ or $i = 2$ with probability $\frac{1- p}{2-p}$.
In this case, $\text{T-Geo}(p,n)$ is identically distributed as $Ber(\frac{1-p}{2-p}) + 1$.  
%our algorithm always returns $1$ as the result. If $n=2$, our algorithm returns $1$ with probability of $
%\frac{1}{2-p}$ and returns $2$ with probability of $\frac{1-p}{2-p}$ according to the definition of $\text{T-Geo}(p, n)$. 
By Fact~\ref{lem:ber2}, this can be achieved 
%by sampling a Bernoulli random variable 
in $O(1)$ expected running time with $O(1)$ worst-case space. 
%according to Fact~\ref{lem:ber2}.
Thus, Theorem~\ref{thm:trun-geo} holds.

\vspace{2mm}
%Next, we consider Case 2 ($n > 2$).
\item[Case~2.1 ($n\geq 3$ and $np \geq 1$):]  
%Our algorithm generates 
In this case, an $i \sim \text{T-Geo}(p, n)$ can be realized by 
repeatedly generating $i \sim \text{B-Geo}(p, n+1)$ until $i \leq n$ holds. 
In each generation, $i \leq n$ 
happens with probability $1 - (1 - p)^n > 1 - \frac{1}{e}$ because $np\geq 1$. 
Thus,~$i$ is realized from $\text{T-Geo}(p, n)$ in 
$O(1)$ trials in expectation.
By Fact~\ref{lmm:bounded-geo}, 
Theorem~\ref{thm:trun-geo} holds.

\vspace{2mm}
\item[Case~2.2 ($n\geq 3$ and $np < 1$):] 
To handle this case, we propose an algorithm which runs as follows:
\begin{itemize}
\item $i \leftarrow 0$;
\item while $i \leq n$ do:
\begin{itemize}
	\item $i \leftarrow  i + \text{B-Geo}(\frac{2}{n}, n + 1)$; 
	\item if $i \leq n$ and $\Ber((1 - p)^{i-1}) = 1$
\begin{itemize}
\item generate $c \leftarrow \Ber(\frac{1}{2 \, p^*})$, with $p^* = \frac{1 - (1-p)^n}{np}$;
\item if $c = 1$, \textbf{return}~$i$ as a realization of a random variate from $\text{T-Geo}(p, n)$; 
\end{itemize}
\end{itemize}	
\item start over from the first step with $i \leftarrow 0$;
\end{itemize}
\end{description}

To complete the proof of Theorem~\ref{thm:trun-geo},
we show the correctness and running time cost of the above algorithm.
%The following claim completes the proof of Theorem~\ref{thm:trun-geo}.
%\begin{claim}\label{claim:trun-geo-proof}
%The above algorithm correctly returns a random variate from $\text{T-Geo}(n, p)$ in $O(1)$ expected time with $O(n)$ worst-case space. 
%\end{claim}

%\begin{proof}
%It remains to prove the correctness and the running time for the case 2 algorithm.

\vspace{2mm}
\noindent
\underline{Correctness.} 
Observe that the $\text{B-Geo}(\frac{2}{n}, n+1)$ generation in the while-loop
simulates the subset sampling process of choosing~$i$, for each index~$i$, independently with probability~$2/n$.
Thus, each~$i$ has probability~$2/n$ of being sampled in this step.
Next,~$i$ is provisionally accepted with probability $(1 - p)^{i-1}$ via a Bernoulli trial, and~$i$ is further eventually accepted with probability $\frac{1}{2p^*}$.
Therefore, combining all these probabilities,
we have:
$$\text{Pr}[ i \text{~is returned as a T-Geo$(p, n)$ variate}] = \frac{2}{n} \cdot (1 - p)^{i-1} \cdot \frac{1}{2 \cdot \frac{1- (1 - p)^n}{np}} = \frac{p (1 - p)^{i-1}}{1 - (1 - p)^n}\,.$$ 
Hence, our algorithm for $np<1$ case correctly realizes a $\text{T-Geo}(p, n)$ random variate.

\vspace{2mm}
\noindent
\underline{The Running Time.} 
First, it can be verified that, $\text{B-Geo}(2/n, n+1)$ returns $i \leq n$ within~$O(1)$ trials in expectation. 
Next, we prove that, given an~$i$, our algorithm accepts~$i$ with at least constant probability. 
To see this, given an~$i$, it is accepted by the first Bernoulli trial with probability $(1 - p)^{i-1} \geq (1 - p)^n \geq (1 - \frac{1}{n})^n \geq \frac{8}{27}$, where the second last inequality comes from the fact that $pn < 1$, while the last inequality is by the fact that $n \geq 3$ and the fucntion $(1 - \frac{1}{n})^n$ monotonically increases with $n$. 
Furthermore, for the second Bernoulli trial, $i$ is accepted with probability $\frac{1}{2 p^*} \geq \frac{1}{2}$, because $p^*  = \frac{1 - (1 - p)^n}{np} \leq 1$. 
Therefore, in expectation, within~$O(1)$
trials in total, an~$i$ will be returned by our algorithm.
Finally, by Facts~\ref{lem:ber2} and~\ref{lmm:bounded-geo}, and by our Theorem~\ref{thm:ber},
all the random variates used in each trial can be generated in~$O(1)$ expected time with $O(n)$~worst-case space. 
%\end{proof}

\section{Our Optimal DPSS Algroithm}
\label{sec:main-algo}

In this section, we show our optimal dynamic parameterized subset sampling algorithm. For ease of exposition, and based on the results in Section~\ref{sec:random}, we note that the evaluation of every random variate used in our algorithm can be performed in~$O(1)$ expected time with~$O(n)$ worst-case space.

\vspace{-1mm}
\subsection{A One-Level Bucket-Grouping Structure}

We first introduce a one-level {\em Bucket-Grouping Structure} (BG-Str) which  
is a key technique to our {\em three-level sampling hierarchy}, and which reveals the intuitions and the correctness of our algorithm.
In the following, we deliberately use $X = \{x_1, \ldots, x_N\}$ to denote the item set, 
to distinguish it from the \emph{original} input item set,~$S$:
the BG-Str will be applied recursively to construct the sampling hierarchy. 
To avoid the complicated {\em floor} and {\em ceiling} operations on the logarithmic values, 
without loss of generality, we assume that the size, $N$, of an item set is always a {\em power-of-$16$} number. 
Otherwise, we can {\em conceptually} pad with {\em dummy} items, each with {\em zero} weight, to make~$N$ a power-of-$16$ number.
This would only increase~$N$ by a constant factor of~$16$ in the worst case, and hence, would not affect our theoretical results.  
%Therefore, in the following, 
Thanks to this, the values of $\log_2 N$, $\log_2 \log_2 N$ and $\log_2 \log_2 \log_2 N$ are all integers.

\vspace{1mm}
\noindent
{\bf The Data Structure.}
Consider a PSS instance $X = \{x_1, \ldots, x_N\}$, the one-level bucket-grouping structure on $X$, denoted by {{BG-Str}($X$)}, can be constructed in four steps:
\begin{itemize}[leftmargin = *]
\item \underline{Step 1:  Total Weight Calculation.} Compute $W_X = \sum_{x \in X} w(x)$. Observe that $W_X$ is just a sum of~$N$ one-word integers, and hence, $W_X$ can be stored in $O(1)$ words.
\item \underline{Step 2: Item Bucketing.}
Assign the items in $X$ to {\em buckets} by their weights.
Bucket $B_X(i)$ is assigned all the items $x \in X$ such that $2^i \leq w(x) < 2^{i+1}$, for $i = 0, 1, \ldots, \lfloor \log_2 w_{\max} \rfloor$, where $w_{\max}$ is the largest possible weight of the items.
All the {\em non-empty} buckets (containing at least one item from $X$) are maintained in a {\em sorted} linked list by their bucket index,~$i$, in ascending order. 

\item \underline{Step 3: Bucket Grouping.} 
Organize all the buckets into {\em groups}, where  
the group with index,~$j$, denoted by~$G_X(j)$, 
is a collection of all the non-empty buckets~$B_X(i)$ whose bucket indices $i$ satisfy $ \left\lfloor \frac{i}{\log_2 N} \right\rfloor = j$, 
where the largest possible value of a group index,~$j$, is $j_{\max} = \left\lfloor \frac{\lfloor \log_2 w_{\max} \rfloor}{\log_2 N} \right\rfloor$.
It is worth mentioning that each group corresponds to exactly $\log_2 N$ consecutive possible bucket indices.
Analogous to the buckets, all the {\em non-empty} groups (containing at least one non-empty bucket) 
are maintained in a {\em sorted} linked list by their group index,~$j$, in ascending order. 

\item \underline{Step 4: Next-Level Instance Construction.} 
For each non-empty group,~$G_X(j)$, 
construct a {\em next-level PSS instance} on a set,~$Y_j$, of 
{\em next-level} items
by the following steps:
\begin{itemize}
\item for each non-empty bucket $B_X(i) \in G_X(j)$, create a {\em next-level item}, denoted by $y_i$, with weight $w(y_i) = 2^{i+1} \cdot |B_X(i)|$.
\item denote the set of these next-level items by $Y_j$, and the PSS instance on $Y_j$ is called the {\em next-level instance} of $G_X(j)$. 
\end{itemize} 

\end{itemize}

\begin{lemma}\label{lmm:one-level-pre}
The one-level bucket-grouping structure, {\em BG-Str($X$)}, can be constructed in $O(N)$ worst-case time with $O(N)$ space consumption.
\end{lemma}
\begin{proof}
The construction of BG-Str($X$)  
mainly involves the constructions and maintenance of {\em sorted linked lists} of non-empty buckets (and non-empty groups) of length at most $\lfloor \log_2 w_{\max} \rfloor \leq d$.
According to Fact~\ref{fact:sorted} in the Word RAM model,
each of them can be constructed in {\em linear} time with {\em linear} space to its size.
As the sum of all their sizes is bounded by~$O(N)$ (because they are non-empty),  
the overall construction cost and space consumption of them are bounded by~$O(N)$.
\end{proof}

\noindent
\underline{\bf Answering a PSS query on~$X$ with BG-Str($X$).}
Consider the bucket-grouping structure $\text{BG-Str}(X)$, constructed on the item set~$X$.
Given a PSS query with parameters $(\a, \b)$, 
the buckets~$B_X(i)$ can be categorized into three types.
It is worth mentioning that BG-Str($X$) itself 
does not have the concept of the query parameters $(\a, \b)$ when constructed. All the following concepts are defined {\em dynamically on the fly} during query time by the given parameters $(\a, \b)$.

\begin{itemize}[leftmargin = *]
\item
{\em $(\a, \b)$-Insignificant Buckets.}
$B_X(i)$ is an {\em $(\a, \b)$-insignificant bucket} if its bucket index $i$ satisfies $\frac{2^{i+1}}{W_X(\a, \b)} \leq \frac{1}{N^2}$, where recall that $W_X(\a, \b) = \a \cdot W_X + \b$.
Since $B_X(i)$ contains items with weights in the range $[2^i, 2^{i+1})$, 
each item $x$ in a non-empty $(\a, \b)$-insignificant bucket 
$B_X(i)$
will have probability $p_x(\a, \b) < \frac{1}{N^2}$ to be selected in the sampling result.

\item
{\em $(\a, \b)$-Certain Buckets.}
$B_X(i)$ is an {\em $(\a, \b)$-certain bucket} 
if its bucket index $i$ satisfies $\frac{2^{i}}{W_X(\a, \b)} \geq 1$. 
As a result, each item in a non-empty $(\a, \b)$-certain bucket $B_X(i)$ is included in the sampling result with probability $1$.

\item
{\em $(\a, \b)$-Significant Buckets.} 
If $B_X(i)$ is neither of the above two types, 
then $B_X(i)$ is an $(\a, \b)$-significant bucket.
By definition, each item in such a non-empty bucket  $B_X(i)$ will have probability in the range of $\left(\frac{1}{2N^2}, 1\right)$ to be sampled.
\end{itemize}

Moreover, 
if all the possible buckets (not necessarily non-empty) in a group $G_X(j)$ 
are $(\a, \b)$-insignificant (resp., $(\a, \b)$-certain), 
then $G_X(j)$ is an {\em $(\a, \b)$-insignificant group} (resp., {\em $(\a, \b)$-certain group}).
If neither of these cases holds,
$G_X(j)$ is a {\em $(\a, \b)$-significant group}. Observe that for every group, the difference between the largest and smallest bucket indices is at most $\log_2 N  - 1$. Thus, no group can simultaneously contain an $(\alpha, \beta)$-certain bucket and an $(\alpha,\beta)$-insignificant bucket.

\begin{lemma}\label{lmm:next-level-num}
Given a PSS query with parameters $(\a, \b)$, there can be at most three
$(\a, \b)$-significant groups in {\em BG-Str($X$)}.
\end{lemma}
\begin{proof}
By definition, each $(\a, \b)$-significant group contains at least one possible $(\a, \b)$-significant bucket.
Furthermore, observe that the union of the probability ranges
of all possible $(\a, \b)$-significant buckets is $(\frac{1}{2N^2}, 1)$.
As a result, there can be at most $1+2 \log_2 N $ such buckets.  
Since each group contains $ \log_2 N $ possible buckets with consecutive bucket indices, therefore, all these $(\a, \b)$-significant buckets must 
be contained in at most three $(\a, \b)$-significant groups.
\end{proof}

\noindent
{\bf The Query Algorithm (Algorithm~\ref{algo:one-level}).}
The basic idea is to answer a PSS query with parameters $(\a, \b)$ on $X$ by solving the following three types of PSS query instances:
\begin{itemize}[leftmargin = *]
\item {\em insignificant instance:} one PSS query with parameters $(\a, \b)$ on the set of all the items from $X$ that are in the buckets in all the $(\a, \b)$-insignificant groups;
\item {\em certain instance:} one PSS query with parameters $(\a, \b)$ on the set of all the items from $X$ that are in the buckets in all the $(\a, \b)$-certain groups;  
\item {\em next-level instances:} one PSS query with parameters $(0, W_X(\a, \b))$ on the next-level item set $Y_j$ for each $(\a, \b)$-significant group $G_X(j)$. As each sampled next-level item in $Y_j$ is essentially a bucket in $G_X(j)$, the items of $X$ in the bucket can be sampled with {\em rejection sampling}.
\end{itemize}

\begin{algorithm}
\caption{One-Level PSS Query Algorithm: Query(BG-Str($X$), $(\a, \b)$, $\ell_{\text{cur}}$)}\label{algo:one-level}
\KwData{a Bucket-Grouping Structure BG-Str($X$), PSS query parameters $(\alpha,\beta)$, the current recursion level $\ell_\text{cur}$}
\KwResult{a PSS result, with parameters $(\a, \b)$, $T\subseteq X$}

$N \leftarrow |X|$\;
$j_1 \leftarrow \left\lfloor \frac{\lfloor \log_2 {\frac{W_X(\a, \b)}{N^2}}\rfloor - 1}{\log_2 N}\right\rfloor$, the maximum possible $(\a, \b)$-insignificant group index\;
$T_{1}\leftarrow$ QueryInsignificant(BG-Str($X$), $(\a, \b)$, $j_1 \cdot \log_2 N$) // Algorithm~\ref{algo:insig}\;
$j_2 \leftarrow \lfloor \frac{\lceil \log_2 W_X(\a, \b) \rceil}{\log_2 N}\rfloor$, the minimum possible $(\a, \b)$-certain group index\;
$T_{2}\leftarrow$ QueryCertain(BG-Str($X$), $(\a, \b)$, $j_2\cdot \log_2 N$)  // Algorithm~\ref{algo:cert}\;
$T_3 \leftarrow \emptyset$\;
\For{$j$ from $j_1 + 1$ to $j_2 -1$}
{$T_3' \leftarrow $ QueryASignificantGroup(BG-Str($X$), $(\a, \b)$, $j$, $\ell_\text{cur}$) //  Algorithm~\ref{algo:nextl}\;
$T_3 \leftarrow T_3 \cup T_3'$\;
}
$T\leftarrow T_1\cup T_2\cup T_3$\;
return $T$\;
\end{algorithm}

\begin{toappendix}
\begin{algorithm}
\caption{QueryInsignificant(BG-Str($X$), $(\a, \b)$, $i_1$)}\label{algo:insig}
\KwData{a Bucket-Group Structure BG-Str($X$), PSS query parameters $(\alpha,\beta)$, the maximum possible bucket index $i_1$ in the $(\a, \b)$-insignificant groups}
\KwResult{a PSS result $T_1 \subseteq X$ within all the $(\a, \b)$-insignificant groups}
initialize $T_{1}\leftarrow \emptyset$ and $N \leftarrow |X|$\; 
$k\leftarrow \text{B-Geo}(1/N^2, N+1)$\;
\If {$k\leq N$}{
$A \leftarrow $ the set of all the items in the non-empty buckets (with bucket index $\leq i_1$)\;
\If{$|A| < k$}{return $\emptyset$\;}
{
$x \leftarrow A[k]$, the $k^\text{th}$ item in $A$\;
$T_1 \leftarrow T_1 \cup \{x\}$ with probability $\frac{p_x(\a, \b)}{1 / N^2}$\;
}
\For{$i \leftarrow \{k + 1, \ldots, |A|\}$}
{
$x \leftarrow A[i]$, the $i^\text{th}$ item in $A$\;
$T_1 \leftarrow T_1 \cup \{x\}$ with probability $p_x(\a, \b)$\;
}
}
return $T_1$\;
\end{algorithm}

\begin{algorithm}
\caption{QueryCertain(BG-Str($X$), $(\a, \b)$, $i_2$)}\label{algo:cert}
\KwData{a Bucket-Group Structure BG-Str($X$), PSS query parameters $(\a, \b)$, the minimum possible bucket index $i_2$ in the $(\a, \b)$-certain groups}
\KwResult{a PSS result $T_{2}\subseteq X$ within all the $(\a, \b)$-certain groups}
initialize $T_{2}\leftarrow \emptyset$\;
\For{each non-empty buckets $B_X(i)$ with bucket index $i \geq i_2$ in BG-Str($X$)}{
$T_{2}\leftarrow T_{2}\cup B_X(i)$\;
}
return $T_2$\;
\end{algorithm}

\begin{algorithm}
\caption{QueryASignificantGroup(BG-Str($X$), $(\a, \b)$, $j$, $\ell_{\text{cur}}$)}\label{algo:nextl}
\KwData{a Bucket-Group Structure BG-Str($X$), PSS query parameters $(\alpha,\beta)$, an $(\a, \b)$-significant group index $j$, the current recursion level $\ell_{\text{cur}}$}
\KwResult{a PSS result $T_{3}\subseteq X$ within the $(\a, \b)$-significant group $G_X(j)$}

\If{$G_X(j)$ is empty}{return $\emptyset$\;}

// obtain a PSS result $T_{Y_j}$ with parameters $(0, W_X(\a, \b))$
on the next-level item set $Y_j$ of $G_X(j)$\;
\If{$\ell_{\text{cur}} = 1$}{
 $T_{Y_j} \leftarrow$ Query(BG-Str($Y_j$), $(0, W_X(\a, \b))$, $\ell_{\text{cur}} + 1$) // recursively invoke Algorithm~\ref{algo:one-level}\;
}
\Else{
$T_{Y_j} \leftarrow$ QueryFinalLevel(BG-Str($Y_j$), $(0, W_X(\a, \b))$) // in Section~\ref{sec:adapter}\;
}

$T_3 \leftarrow$ ExtractItems(BG-Str($X$), $T_{Y_j}$, $(0, W_X(\a, \b))$) // Algorithm~\ref{algo:extract}\;

return $T_3$\;
\end{algorithm}

\begin{algorithm}
\caption{ExtractItems(
BG-Str($X$), $T_{Y_j}$, $(0, W_X(\a, \b))$)}
\label{algo:extract}

\KwData{a Bucket-Group Structure BG-Str($X$), 
a PSS result $T_j$ on the next-level items, each corresponds to a candidate bucket of items of $X$, 
PSS query parameters $(0, W_X(\alpha,\beta))$ 
}
\KwResult{a PSS result $T_3 \subseteq X$ 
from the items in the candidate bucket list $T_{Y_j}$}

initialize $T_3 \leftarrow \emptyset$\;
\For{each bucket $B_X(i) \in T_{Y_j}$}
{
$p\leftarrow \min\{1,\frac{2^{i+1}}{W_X(\alpha,\beta)}\}$ and $n_i \leftarrow |B_X(i)|$\; 
\If {$p\cdot n_i\geq 1$}{
$k\leftarrow \text{B-Geo}(p,n_i + 1)$\;
}
\Else{
\If {$\Ber(\frac{1-(1-p)^{n_i}}{p\cdot n_i})=0$}{
continue to process next bucket in $T_{Y_j}$\;
}
$k \leftarrow \text{T-Geo}(p, n_i)$\;
}
\While {$k\leq n_i$}{
$x \leftarrow $ the $k^\text{th}$ item in $B_X(i)$\;
$T_{3}\leftarrow T_{3}\cup\{x\}$ with probability $\frac{p_x(\alpha,\beta)}{p}$\;\label{alg:extract:Bernoulli}
$k\leftarrow k+\text{B-Geo}(p,n_i+1)$\; 
}
}
return $T_3$\;
\end{algorithm}

\end{toappendix}

The detailed steps of the query algorithm are shown in Algorithm~\ref{algo:one-level}.

The query algorithm handles these three different types of instances 
with Algorithms~\ref{algo:insig}, ~\ref{algo:cert}, and~\ref{algo:nextl}, respectively.
The main methodology in our algorithm is {\em rejection sampling}. 
Specifically, for each item~$x$ with sampling probability~$p_x(\a,\b)$,
we first sample~$x$ as a {\em potential item} with some probability $p' \geq p_x(\a, \b)$, and then, we accept this item~$x$ 
as a PSS result with probability~$\frac{p_x(\a, \b)}{p'}$ independently.
In this way, item~$x$ will be sampled
with probability~$p_x(\a, \b)$ independently as required in the PSS query.

Next, we claim that that all the mathematical computations involved in Algorithm~\ref{algo:one-level} can be evaluated in $O(1)$ time in the Word RAM model. 
To prove this, it suffices to show the following Claim~\ref{claim:logcalculation}. 
%With this claim, 
As a result, the largest index $j_1$ of the insignificant groups and the smallest index $j_2$ of the certain groups can be computed in $O(1)$ time in the Word RAM model, as $W_X(\alpha,\beta)$ is a rational number that satisfies the conditions in Claim~\ref{claim:logcalculation}.

\begin{claimapprep}\label{claim:logcalculation}
Consider a rational number $x>0$ with its numerator and denominator being $O(1)$-word integers, then $\lceil \log_2 (x) \rceil$ and $\lfloor \log_2(x) \rfloor$ can be calculated in $O(1)$ time. 
\end{claimapprep}
\begin{proof}
Denote the numerator and the denominator of $x$ by $A$ and $B$, respectively, i.e., $x = A / B$. Both of them fit in $O(1)$ words.
In the Word RAM model, $\lceil\log_2(A)\rceil$ (and $\lceil\log_2(B)\rceil$), can be calculated with bit-wise operation by finding the index of the highest non-zero bit.
%According to the logarithmic rules of computation,
Observe that $\log_2 (x) = \log_2 (A)-\log_2(B) $. 
Thus, $\lceil \log_2 (x) \rceil<\log_2(x)+1< \lceil\log_2(A)\rceil-\lceil\log_2(B)\rceil+2$ and $\lceil \log_2 (x) \rceil\geq \log_2(x)>\lceil\log_2(A)\rceil-1-\lceil\log_2(B)\rceil$. 
Specifically, $c_1 -1 \leq \lceil \log_2 (x) \rceil \leq c_1$, where
%Therefore, the value of $\lceil \log_2 (x) \rceil$ either equals to 
$c_1=\lceil\log_2(A)\rceil-\lceil\log_2(B)\rceil+1$.
In other words, $\lceil \log_2 (x) \rceil$ is equal to either $c_1 -1 $ or $c_1$.
%or equals to $c_2=\lceil\log_2(A)\rceil-\lceil\log_2(B)\rceil$. 

As a result, it suffices to check whether $2^{c_1 - 1} < x \leq 2^{c_1}$, equivalently, $2^{c_1 - 1} \cdot B < A \leq 2^{c_1} \cdot B$ holds or not.
If this is the case, then $\lceil \log_2 (x) \rceil = c_1$, otherwise, $\lceil \log_2 (x) \rceil = c_1 - 1$.

Since $c_1$ is bounded by $O(d)$, the value of $2^{c_1}$ can be represented by $O(1)$ words and it can be computed by bit shifting operation in $O(1)$ time. 
Hence, we can decide 
the value of $\lceil \log_2 (x) \rceil$ in $O(1)$ time.
%
%The remaining step is to check whether $2^{c_1}\cdot B\geq A$ (equivalently $B\geq A\cdot 2^{-c_1}$) and $A>2^{c_1-1}\cdot B$ (equivalently $B<A\cdot 2^{1-c_1}$). Observe that power of $2$ with non-negative exponent can be computed using shift operations. If both inequalities hold true, the algorithm outputs $c_1$ as the value of $\lceil \log_2 (x) \rceil$; otherwise, it outputs $c_2$. Clearly, each steps of the above algorithm can be evaluated in the Word RAM model in $O(1)$ time. 

The calculation of 
%The techniques to calculate 
$\lfloor \log_2(x) \rfloor$ is analogous and hence, the details are omitted here.
\end{proof}

\vspace{1mm}
\noindent
\underline{\em Handling the Insignificant Instance (Algorithm~\ref{algo:insig}).}
By definition, each item~$x$ in the buckets in all those $(\a, \b)$-insignificant groups has probability at most~$\frac{1}{N^2}$ to be sampled.
Thus, we first generate a {\em bounded geometric variate} $k \sim \text{B-Geo}(\frac{1}{N^2}, N+1)$ to 
simulate the process of flipping a Bernoulli coin $\Ber(\frac{1}{N^2})$ for each item, one by one, and~$k$ is the index of the first item to be sampled. 
Clearly, if $k > N$, then nothing will be sampled and thus done. Otherwise, Algorithm~\ref{algo:insig} first collects all the items in this insignificant instance and stores them in an array~$A$. 
If $k > |A|$, then it is done. 
But if~$k \leq |A|$, the item~$x$ at the $k^\text{th}$ index in $A$, i.e., $A[k]$, is a potential item, and it is accepted in the sampling result with probability $\frac{p_x(\a, \b)}{1/ N^2}$.
Thus, this item~$x$ is equivalently sampled with probability $p_x(\a, \b)$.
For each of the remaining items~$x$, at index starting from~$k+1$ in the array $A$, Algorithm~\ref{algo:insig} flips a Bernoulli coin with probability~$p_x(\a, \b)$ to accept~$x$.
Therefore, each item~$x$ in the array~$A$ will be correctly sampled, independently, with probability~$p_x(\a, \b)$.

In terms of query efficiency, 
recall that the bounded geometric variate $k \sim \text{B-Geo}(\frac{1}{N^2}, N+1)$ can be generated in $O(1)$ expected time.
If $k > N$, the running time cost is~$O(1)$.
Otherwise, the process takes~$O(N)$ time.
So the overall expected running time is bounded, thus leading to Lemma~\ref{lmm:algo2}.
\begin{align*}
\Pr[k\leq N]\cdot O(N)+\Pr[k> N]\cdot O(1) 
\leq \left(1-\left(1-\frac{1}{N^2}\right)^{N}\right)\cdot O(N)+O(1)
\leq  \frac{O(N)}{N}+ O(1)
 =O(1)\,. 
\end{align*}
\begin{lemma}\label{lmm:algo2}
Algorithm~\ref{algo:insig} returns a correct PSS result with parameters $(\a, \b)$ on the set of items in the insignificant instance in~$O(1)$ expected time.
\end{lemma}

\vspace{1mm}
\noindent
\underline{\em Handling the Certain Instance (Algorithm~\ref{algo:cert}). }
For the certain instance, each item $x$ has a sampling probability $p_x(\a, \b) = 1$. Therefore, Algorithm~\ref{algo:cert} simply collects and returns all these items in the corresponding buckets, giving the following lemma.
\begin{lemma}\label{lmm:algo3}
Algorithm~\ref{algo:cert} returns a correct PSS result with parameters $(\a, \b)$ on the set of items in the certain instance in $O(1 + \ell)$ time, where~$\ell$ is the number of items sampled.
\end{lemma}

\vspace{1mm}
\noindent
\underline{\em Handling Each Significant Group (Algorithm~\ref{algo:nextl}).}
We now discuss the details for handling each $(\a, \b)$-significant group $G_X(j)$.
The core idea of Algorithm~\ref{algo:nextl} is to first obtain a subset $T_{Y_j}$ of the next-level item set  $Y_j$ such that each next-level item $y_i \in Y_j$ is sampled  independently with probability  
$\min\left\{1, \frac{w(y_i)}{W_X(\a, \b)}\right\}$.
To achieve this, Algorithm~\ref{algo:nextl} invokes a PSS query with parameters $(0, W_X(\a, \b))$ on $Y_j$. 
As each next-level item $y_i$ corresponds to a bucket $B_X(i)$, $T_{Y_j}$ is treated as a list of the buckets corresponding to its next-level items.
And each bucket $B_X(i) \in T_{y_j}$ is called a {\em candidate bucket} because it has a chance to contain a {\em potential item} $x\in B_X(i)$ which is sampled with probability of $p = \min\left\{1, \frac{2^{i+1}}{W_X(\a, \b)}\right\}$. 
If a candidate bucket $B_X(i)$ is confirmed via rejection sampling to have at least a potential item in it, then $B_X(i)$ is called a {\em promising bucket}.

The next step is to extract items from the candidate bucket. 
As this step will be reused later, we summarize the steps in Algorithm~\ref{algo:extract}.
Specifically, Algorithm~\ref{algo:extract} accepts each candidate bucket $B_X(i) \in T_{Y_j}$ as a promising bucket by the following steps, depending on two cases. Denote by $n_i = |B_X(i)|$ the number of items in $B_X(i)$.
\begin{description}[leftmargin = *]
\item[Case 1 ($p \cdot n_i \geq 1$):]
In this case, 
bucket $B_X(i)$ indeed has probability of $1$ to be included in $T_{Y_j}$, i.e., $\text{Pr}[B_X(i) \in T_{Y_j}] = 1$.
Algorithm~\ref{algo:extract} generates a random variate $k \sim \text{B-Geo}(p, n_i + 1)$. 
If $k > n_1$, which happens with probability $q = (1 - p)^{n_i}$, 
then reject~$B_X(i)$ because $\text{Pr}[B_X(i) \in T_{Y_j}] \cdot q = (1-p)^{n_i}$ is exactly the probability that~$B_X(i)$ has no potential items. 
Otherwise, $B_X(i)$ is accepted as a promising bucket and the $k^\text{th}$ item $x$ in $B_X(i)$
is a potential item sampled with probability $p$.

\vspace{1mm}
\item[Case 2 ($p \cdot n_i < 1$):]
In this case, $\text{Pr}[B_X(i) \in T_{Y_j}] = p \cdot n_i$.
Algorithm~\ref{algo:extract} generates a Bernoulli coin $c \sim \text{Ber}\left(\frac{1 - (1 - p)^{n_i}}{p \cdot n_i}\right)$.
If $c = 0$, then reject~$B_X(i)$.
Otherwise,~$B_X(i)$ is accepted 
as a promising bucket, observing that this event happens with probability  $\text{Pr}[B_X(i) \in T_{Y_j}] \cdot \text{Pr}[c=1] = 1 - (1 - p)^{n_i}$, which is exactly the probability of $B_X(i)$ containing at least one potential item.
However, the most tricky aspect is to find the first potential item,~$x$,
in~$B_X(i)$ conditioned on the fact that~$B_X(i)$ is a promising bucket, such that~$x$ is still equivalently sampled with probability of $p$. 
To overcome this, Algorithm~\ref{algo:extract} generates a {\em truncated geometric} random variate $k\sim \text{T-Geo}(p, n_i)$, where the truncated geometric distribution guarantees to return an index $k$ such that (i) $k \leq n_i$, and (ii) $\text{Pr}[k = j] = \frac{p (1 - p)^{j-1}}{1 - (1 - p)^{n_i}}$.
As discussed in Section~\ref{sec:random},
the value of $k$ has the same distribution as the smallest sampled index in $[1, n_i]$, each sampled with probability $p$, conditioned on that at least one index being sampled. 
Therefore, the $k^\text{th}$ item in $B_X(i)$ has a probability of~$p$ to be sampled as a potential item.
\end{description}

As a result, in either case, if~$B_X(i)$ is confirmed as a promising bucket, Algorithm~\ref{algo:extract} always gives the index~$k$ 
of the first potential item~$x$ in~$B_X(i)$ with correct probability,~$p$.
Algorithm~\ref{algo:extract} 
further generates a Bernoulli coin $c \sim \text{Ber}\left(\frac{p_x(\a, \b)}{p}\right)$ (Line~\ref{alg:extract:Bernoulli}) to accept~$x$ and add it to the PSS result,~$T_3$. 
Then the algorithm repeatedly generates the index of the next potential item with bounded geometric distribution $\text{B-Geo}(p, n_i+1)$ and performs the same rejection sampling process 
until the index is $> n_i$.

\begin{lemma}\label{lmm:algo4}
Algorithm~\ref{algo:nextl} returns a correct PSS result with parameters $(\a, \b)$ on the set of items in the given $(\a, \b$)-signifincant group $G_X(j)$. 
\end{lemma}
\begin{proof}
The correctness of Algorithm~\ref{algo:nextl} follows from the discussion above. 
That is, each bucket $B_X(i) \in Y_j$ has a probability of $1 - (1 - p)^{n_i}$ to be accepted as a promising bucket independently;
and each item in a promising bucket $B_X(i)$ has correct probability~$p$ to be identified as a potential item independently, where $n_i = |B_X(i)|$ and $p = \min\{1, \frac{2^{i+1}}{W_X(\a, \b)}\}$. 
Furthermore, each potential item,~$x$, is accepted in the returned sampling result independently with probability~$\frac{p_x(\a, \b)}{p}$. 
Therefore, each item is sampled in the returned sampling result independently with probability~$p_x(\a, \b)$. 
\end{proof}

\begin{theorem}[Correctness]\label{lmm:one-level-correctness}
Given a PSS query with parameters $(\a, \b)$ on $X$, 
let $T \subseteq X$ be the sampling result returned by Algorithm~\ref{algo:one-level}.
Every item $x \in X$ is selected in $T$ independently with probability $p_x(\a, \b)$. 
\end{theorem}
\begin{proof}
This theorem follows immediately from Lemmas~\ref{lmm:algo2}, ~\ref{lmm:algo3} and~\ref{lmm:algo4}.
\end{proof}

\begin{theorem}[Query Complexity]\label{lmm:one-level-query}
Suppose that every PSS query with parameters $(0, W_X(\a, \b))$ on the next-level item set $Y_j$ of a $(\a, \b)$-significant group can be answered in $O(1 + \mu_{Y_j})$ time in expectation, where $\mu_{Y_j}$ is the expected size of the sampling result on $Y_j$.
The overall expected running time of Algorithm~\ref{algo:one-level} is bounded by $O(1 + \mu_X(\a,\b))$. 
\end{theorem}

\begin{proof}
Observe that Algorithm~\ref{algo:one-level} invokes Algorithm~\ref{algo:insig} and Algorithm~\ref{algo:cert} each once, and Algorithm~\ref{algo:nextl} at most three times according to Lemma~\ref{lmm:next-level-num}.
Moreover, by Lemma~\ref{lmm:algo2},
the expected running time of Algorithm~\ref{algo:insig} for the insignificant instance is $O(1)$.  
And by Lemma~\ref{lmm:algo3}, 
Algorithm~\ref{algo:cert} takes $O(1 + \ell)$ time, where $\ell$ is the number of items in the certain instance.
Therefore, it suffices to show that the expected running time of Algorithm~\ref{algo:nextl} for each significant group $G_X(j)$ is bounded by $O(1 + \mu)$, where $\mu$ is the expected size of the PSS result
on the items in $G_X(j)$ with parameters $(\a, \b)$.

First, by the assumption in the theorem statement, the list $T_{Y_j}$ of candidate buckets can be obtained in $O(1 + \mu_{Y_j})$ expected time, where $\mu_{Y_j}$ is equivalent to the expected number of candidate buckets, as discussed earlier. 
Consider the following claim.
\begin{claim}\label{claim:1}
Each candidate bucket $B_X(i) \in T_{Y_j}$ has~$\Omega(1)$ probability to be accepted as a promising.
\end{claim}
If Claim~\ref{claim:1} holds, then $\mu_{Y_j}$ can be bounded by $O(\mu_{\text{prom}})$, where $\mu_{\text{prom}}$ is the expected number of promising buckets. 
Moreover, by definition, 
each promising bucket $B_X(i)$  contains at least one potential item sampled with probability $p = \min\{1, \frac{2^{i+1}}{W_X(\a, \b)}\}$.
Thus, $\mu_{\text{prom}}$ 
is bounded by the expected number of potential items.
Furthermore, each potential item $x \in B_X(i)$ is accepted to be in the sampling result with probability $\frac{p_X(\a, \b)}{p} \geq \frac{1}{2}$.
This is because, by the fact that $x \in B_X(i)$, we have $w(x) \in [2^i, 2^{i+1})$.  
Therefore, $\mu_{Y_j}$ is then bounded by 
the expected number of promising buckets,
then bounded by the expected number of potential items,
and thus, bounded by the expected size of the PSS result, $O(\mu)$.
And hence, all the processing costs on the candidate buckets, promising buckets and the potential items
can be charged to the output cost, in expectation. 

Now, we \textbf{prove} Claim~\ref{claim:1}.
Consider a candidate bucket $B_X(i) \in T_{Y_j}$; let $n_i = |B_X(i)|$ and $p = \min\{1, \frac{2^{i+1}}{W_X(\a, \b)}\}$.
Recall that Algorithm~\ref{algo:nextl} considers two cases:
\begin{description}[leftmargin = *]
\item[Case 1 ($p \cdot n_i \geq 1$):] 
In this case, $B_X(i)$ is rejected with probability $(1 - p)^{n_i} \leq (1 - \frac{1}{n_i})^{n_i} \leq \frac{1}{e}$.
Thus,~$B_X(i)$ is accepted as a promising bucket with probability at least $1 - \frac{1}{e} \in \Omega(1)$.

\item[Case 2 ($p \cdot n_i < 1$):] 
Recall that, in this case, $B_X(i)$ is accepted with probability $\frac{1 - (1 - p)^{n_i}}{p \cdot n_i}$.
We claim that this probability is at least $1 - \frac{1}{e}$. 
To see this, consider the function 
$f(a ,b)=\frac{1-(1-b)^{a/b}}{a}$ for $a = p\cdot n_i < 1$ and $b=p \in (0, 1)$.
It can be verified that $f(a, b)$ decreases monotonically with respect to $a$. 
To see this, the partial derivative of $f$ with respect to $a$ is computed as $\frac{\partial f}{\partial a}=\frac{-\frac{a}{b}\cdot (1-b)^{a/b}\cdot \log (1-b)+(1-b)^{a/b}-1}{a^2}$.
Thus, it suffices to consider the sign of the numerator of $\frac{\partial f}{\partial a}$.

Let $c=(1-b)^{1/b}$ and define $g(x)=-x \cdot c^x \cdot \log c +c^x-1=c^x(1-x\log c)-1$. 
Clearly, $g(0)=0$ and $g'(x)=-(\log c) \cdot c^x+c^x \cdot \log c \cdot  (1-x\log c)=-x\cdot c^x\cdot (\log c)^2<0 $ for $x>0$.
This implies that $g(x)<0$ for $x\in (0,1)$ and hence, $\frac{\partial f}{\partial a}<0$ for $a\in (0,1)$.
Thus, we have $f(a,b)> 1-(1-b)^{1/b} > 1- \frac{1}{e} \in \Omega(1)$.
As a result, $B_X(i)$ is accepted as a promising bucket with probability at least~$\Omega(1)$.
\end{description}
Combining the above two cases, Claim~\ref{claim:1} holds.
Moreover, our results on random variate generation in Section~\ref{sec:random} imply that all the random variates used in Algorithm~\ref{algo:nextl} can be generated in $O(1)$ expected time.  
Putting all the above together, Theorem~\ref{lmm:one-level-query} follows.
\end{proof}

Theorem~\ref{lmm:one-level-query} implies that, as long as we can answer the PSS queries optimally 
on the next-level instances, we can answer the PSS query on~$X$ optimally.
Moreover, observe that the next-level item set size is at most~$\log_2 N$, logarithmic in the size of~$X$. 
This motivates an idea of recursively applying the same bucket-grouping technique on next-level item set~$Y_j$ of each group~$G_X(j)$. 
When the next-level instance 
is ``small'' enough, we use a {\em lookup table} to answer the PSS queries on them optimally.
While this idea looks natural,
certain non-trivial technical challenges remain to be addressed.
First, by Lemma~\ref{lmm:next-level-num}, each recursion level increases the number of next-level instances by a factor of~$3$. 
Thus, the recursion cannot be so deep that it exceeds~$O(1)$ recursion levels.
Otherwise, a super-constant factor would be introduced to the query time complexity. 
Second, 
having only~$O(1)$ recursion levels does not reduce the problem size to~$O(1)$. 
Thus, it is not trivial to design a lookup table with space consumption bounded by~$O(N)$. 
Third, the next-level instances are {\em dynamic}, depending on the query parameters $(\a, \b)$. 
However, the lookup table is {\em static}. 
As a result, reducing the next-level instances to problems where the lookup table is applicable during the query time is a challenging task.

\subsection{The Three-Level Sampling Hierarchy}
Given a set $S$ of $n$ items, the three-level sampling hierarchy of $S$ is constructed as follows:
\begin{itemize}[leftmargin = *]
\item \textbf{Level-1}: 
Construct BG-Str($S$);

\item \textbf{Level-2}:
For each group $G_S(j)$ in BG-Str($S$):
\begin{itemize}
\item construct BG-Str($Y_j$) on the next-level item set $Y_j$ of $G_S(j)$;
\item denote each bucket with bucket index $i$ by $B_{Y_j}(i)$, and 
each group with group index $k$ by~$G_{Y_j}(k)$;
\item denote the next-level item corresponding to bucket~$B_{Y_j}(i)$ 
by~$z_i$ and the next-level item set of group~$G_{Y_j}(k)$ by~$Z_k$.

\item \textbf{Level-3}: 
For each group $G_{Y_j}(k)$ in BG-Str($Y_j$),
\begin{itemize}
\item construct BG-Str($Z_k$) on the next level item set $Z_k$ of $G_{Y_j}(k)$;
\item denote each bucket with bucket index $i$ by $B_{Z_k}(i)$ and its corresponding next-level item by~$v_i$;
\item denote the set of all the next-level items of the buckets by~$V_k$; the PSS instance on~$Z_k$ is called a {\em final-level instance}.
\end{itemize}
\end{itemize}
\end{itemize}

\begin{figure}[t]
    \hspace*{-4mm}\includegraphics[width=1\textwidth]{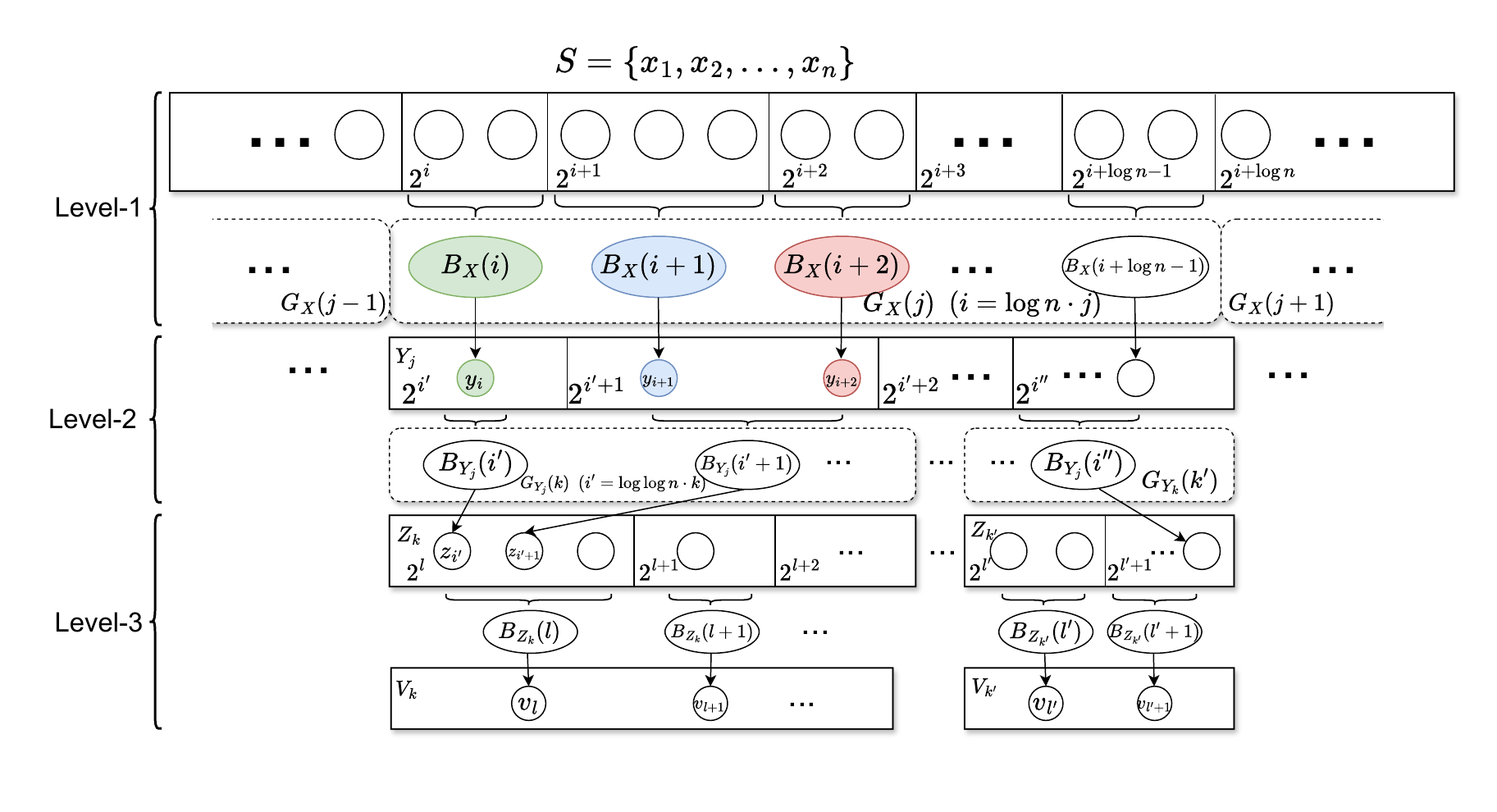}
\vspace{-6mm}
    \caption{Visualization of the three-level sampling hierarchy. Here, circles represent items and ovals represent buckets. 
The diagram illustrates the structure of a level-1 group $G_X(j)$, where each bucket in $B_X(i)$ corresponds to an item $y_i\in Y_j$. 
The BG-Str($Y_j$) contains multiple buckets in multiple groups. 
Each group, e.g. $G_{Y_j}(k)$ corresponds to a next-level item set, e.g. $Z_k$. 
In level-3, each item set $Z_k$ generates one final-level instance $V_k$, in which each item corresponds to a bucket in level-3.}
    \label{fig:hierarchy}
\end{figure}

\begin{lemma}
Given a set $S$ of $n$ items, the three-level sampling hierarchy can be constructed in $O(n)$ time with $O(n)$ space consumption.
\end{lemma}
\begin{proof}
This lemma follows from Lemma~\ref{lmm:one-level-pre} and the facts that, (i) at each level, the total number of the (next-level) items is at most $n$, and (ii) there are only three levels.
\end{proof}

\noindent
{\bf The Overall Query Algorithm.}
To perform a PSS query with parameters $(\a, \b)$ on $S$, the overall query algorithm 
is as simple as invoking Algorithm~\ref{algo:one-level}, Query(BG-Str($S$), $(\a, \b)$, $\ell_{\text{cur}} = 1$).

\begin{lemma}[Correctness]
Every PSS query with parameters $(\a, \b)$ on $S$ can be correctly answered with the three-level sampling hierarchy.
\end{lemma}
\begin{proof}
Follows immediately by applying Theorem~\ref{lmm:one-level-correctness} inductively.
\end{proof}

\begin{lemma}[Query Complexity]\label{lmm:hierarchy-query-time}
If there exists an algorithm that can solve every final-level instance on $Z_k$ optimally in $O(1 + \mu_{Z_k})$ expected time, then the PSS query with parameters $(\a, \b)$ on $S$ can be answered optimally in $O(1 + \mu_S(\a, \b))$ expected time.
\end{lemma}
\begin{proof}
First, by Lemma~\ref{lmm:next-level-num}, there are at most $3$ next-level instances on $Y_j$'s at level-1.
By the same lemma, each of these instances on $Y_j$ can produce at most $3$ next-level instances on $Z_k$'s  at level-2.
And each instance on $Z_k$ corresponds to a final-level instance $V_k$.
Therefore, there are in total at most~$9$ final-level instances.

Second, by mathematical induction from the final level, with Theorem~\ref{lmm:one-level-query}, the query complexity of the one-level algorithm, and by the fact that there are at most~$9$ final-level instances,
the overall expected query time complexity follows.
\end{proof}

Thus, the problem reduces to solving the final-level PSS instances optimally.
Our approach is to use a lookup table.
However, the basic idea of a lookup table is 
to {\em hard code} all the possible sampling results for all possible {\em input configurations} for a special family of {\em static subset sampling} problems.  
Therefore, three main challenges must be addressed.
The first challenge is to {\em efficiently} construct a suitable static subset sampling instance from a final-level PSS instance, where the latter is dynamic and parameterized on the fly by $(\a, \b)$. 
The second challenge is how to ``encode'' the configuration of static subset sampling instance such that it can be evaluated and located in the lookup table in $O(1)$ time.
Third, the overall space consumption of both the lookup table and the ``encoder'' must be bounded by $O(n)$. 
Next, we introduce the remaining two components in our HALT structure, i.e., the {\em adapters} and the {\em lookup table}, to address these technical challenges. 

\subsection{The Lookup Table in HALT}
\label{sec:LT}

In order to set up the context of the goal which the adapters aim to achieve, 
we first introduce the lookup table in our HALT structure.

Define $n_0$ be the size of the item set $S$ when the HALT structure is constructed.
Our lookup table aims to solve a {\em special static subset sampling} (4S) problem in the following form:
\begin{itemize}[leftmargin = *]
\item the input item set $V$ contains {\em exactly} $K$ items $\{v_1, v_2, \ldots, v_K\}$, where $K  = 2 \cdot \log_2 \log_2 \log_2 n_0 $;
\item the sampling probability of the $i^{\text{th}}$ item~$v_i$ in~$V$ is in the form of $p_i = \min \{1, \frac{2^{i+1} \cdot c_i}{m^2} \}$, where~$c_i$ is an integer in the range of~$[0, m]$ and~$m$ is an integer bounded by $O(\log \log n_0)$; as a result, the probability of each item must be an integer multiple of~$\frac{1}{m^2}$. 
\end{itemize}
Therefore, 
every possible input to the above 4S problem can be uniquely specified by a $K$-dimensional vector $\vec{c} = (c_1, c_2, \ldots, c_K)$, where each $c_i$ is an integer in $[0, m]$.  
Each $\vec{c}$ is called an {\em input configuration} of the 4S problem.  
The total number of possible input configurations is $(m + 1)^K$. 

Consider a fixed input configuration $\vec{c}$.
Every possible subset sampling result can be encoded by a $K$-bit string $r$, where the $i^{\text{th}}$ bit $r[i] = 1$ indicates the item $v_i$ is in the sampling result and $r[i] = 0$ means $v_i$ is not.
Thus, the probability of obtaining a subset sampling result $r$ can be calculated as:
$$\text{Pr}(r) = \prod_{i = 1}^K \left(r[i]\cdot p_i + (1 - r[i]) \cdot (1 - p_i)\right)\,.$$
Since $p_i$ is an integer multiple of $\frac{1}{m^2}$, $\text{Pr}(r)$ is an integer multiple of $\frac{1}{(m^2)^K}$.
Furthermore, observe that there are exactly~$2^K$ possible subset sampling results,~$r$. 
Summing $\text{Pr}(r)$ over all $r$'s,  we must have $\sum_{r \in \{0, 1\}^K} \text{Pr}(r) = 1$.

Therefore, all the possible subset sampling results for a fixed input configuration,~$\vec{c}$, can be encoded in an array $A(\vec{c})$ of $(m^2)^K$ cells and each cell is an $K$-bit string.
Specifically, each subset sampling result $r$ is stored in exactly $\text{Pr}(r) \cdot (m^2)^K$ cells in the array $A(\vec{c})$.  
As a result, 
given an input configuration $\vec{c}$ to the 4S problem, 
a subset sampling result for $\vec{c}$ can be obtained 
by picking a cell in $A(\vec{c})$ uniformly at random and returning 
the $K$-bit string $r$ as the subset sampling result. 
Clearly, each subset sampling result $r$ will be returned with probability exactly $\text{Pr}(r)$ independently.

Figure~\ref{fig:table} shows the detailed structure of the lookup table.
\begin{figure}[tbp]
    \centering
    \includegraphics[width=\textwidth]{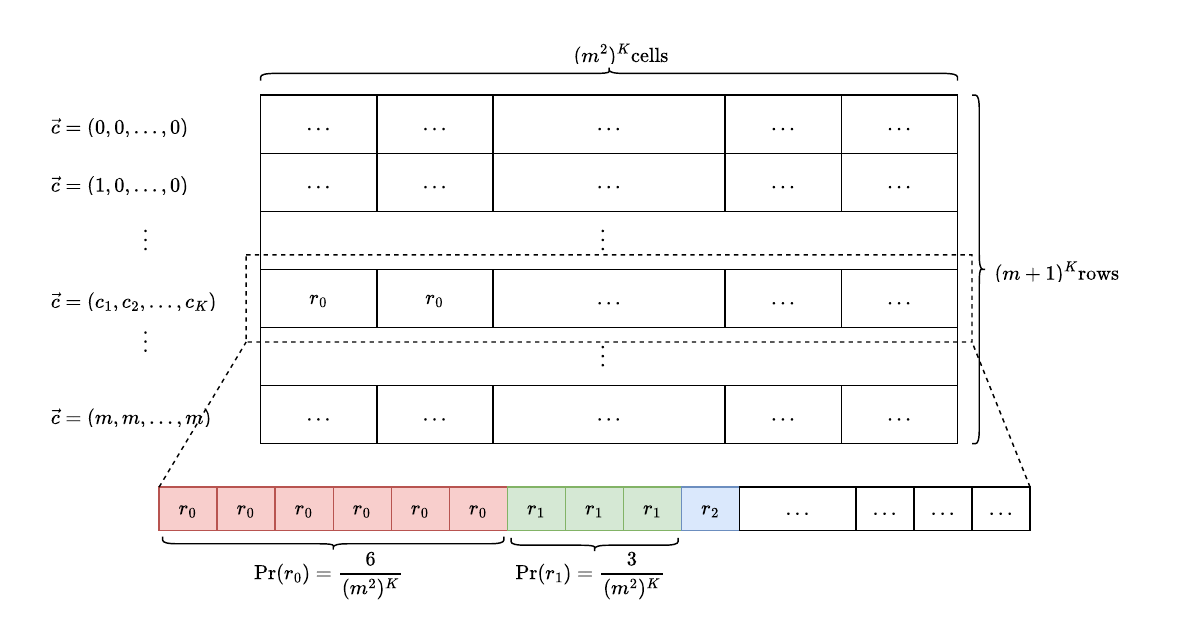}
\vspace{-4mm}
    \caption{Visualization of a lookup table. 
It contains $(m+1)^K$ rows, each row corresponds to a configuration, and each row contains $(m^2)^K$ cells. 
Each cell holds a $K$-bit string, e.g. $r_0,r_1,\ldots$. 
As shown in the figure, for a given configuration $\vec{c}=\{c_1,c_2,\ldots,c_K\}$, if $\text{Pr}(r_0)=\frac{6}{(m^2)^K}$, then there are 6 cells storing $r_0$ in the corresponding row.\vspace{-5mm}}
    \label{fig:table}
\end{figure}

\vspace{2mm}
\noindent{\bf The Lookup Table and Query Algorithm.}
Our lookup table $\mathcal{T}$ contains $(m + 1)^K$ rows, where each row is an array $A(\vec{c})$ corresponding to an input configuration $\vec{c}$.
Given an input configuration $\vec{c}$,
the query algorithm works as follows:
\begin{itemize}[leftmargin = *]
\item locate the corresponding row $A(\vec{c})$ in $\mathcal{T}$ of $\vec{c}$;
\item uniformly at random select an index in $A(\vec{c})$ and return the $K$-bit string in the cell at the index as the subset sampling result.  
\end{itemize}

\vspace{-1mm}
\begin{lemma}
Each input configuration~$\vec{c}$ can be represented in $O(1)$ words.
\end{lemma}
\begin{proof}
Since each~$c_i$ is an integer in $[0, m]$, it can be encoded in $\lceil \log_2 (m+1)\rceil $ bits. Thus, $\vec{c}$ can be encoded in $K \cdot \lceil \log_2 (m+1)\rceil \in O( \log n_0) \subseteq O(d)$ bits, equivalently, $O(1)$ words, where $K \in O(\log \log \log n_0)$ and $m \in O(\log \log n_0)$.  
\end{proof}

\vspace{-1mm}
\begin{lemma}
The lookup table~$\mathcal{T}$  answers each subset sampling query for each configuration~$\vec{c}$ for the 4S problem in $O(1 + \mu(\vec{c}))$ expected time, where~$\mu(\vec{c})$ is the expected size of the subset sample.
\end{lemma}
\begin{proof}
Since each~$\vec{c}$ can be encoded in $O(1)$ words, locating the row~$A(\vec{c})$ in the lookup table~$\mathcal{T}$ can be performed in~$O(1)$ time.
By our assumption at the beginning of this section, generating an index of~$A(\vec{c})$ uniformly at random
can be achieved in~$O(1)$ expected time.
Retrieving the subset sampling result from the $K$-bit string~$r$ can be performed in $O(1 + |r|)$ time in the Word RAM model, where $|r|$ is the number of items in the sampling result. 
Therefore, the overall query time is bounded by~$O(1 + \mu(\vec{c}))$ in expectation.
\end{proof}

\vspace{-1mm}
\begin{lemma}
The lookup table~$\mathcal{T}$ can be constructed in~$O(n_0)$ time and the space consumption of~$\mathcal{T}$ is bounded by~$O(n_0)$. 
\end{lemma}

\begin{proof}
For an input configuration~$\vec{c}$, 
there are~$2^K$ possible subset sampling results~$r$. The probability of each~$r$ can be computed in~$O(K)$ time. 
Furthermore, $A(\vec{c})$ contains $m^{2K}$ cells and each cell stores $K$ bits. The total running time for constructing $A(\vec{c})$ is thus bounded by $O(K \cdot 2^K + K \cdot m^{2K}) = O(K \cdot m^{2K})$ time.
Summing over all $(m + 1)^K$ rows, the overall construction time of $\mathcal{T}$ is bounded by 
$O( (m + 1)^K \cdot K \cdot m^{2K})$ time.
As $K \in O(\log \log \log n_0)$ and $m \in O(\log \log n_0)$, this construction time is thus bounded by $O(n_0)$.
In terms of space consumption, $\mathcal{T}$ contains in total $(m + 1)^K \cdot K \cdot m^{2K}$ bits, which is bounded by $O(n_0)$ words. 
\end{proof}

\vspace{-2mm}
\subsection{Bridging the Hierarchy and the Lookup Table via Adapters}
\label{sec:adapter}

It remains to show how to convert a final-level PSS query instance into an input configuration for the 4S problem in $O(1)$ time on the fly during query time. 
Once this is accomplished, then Lemma~\ref{lmm:hierarchy-query-time} implies that the PSS query $(\a, \b)$ on $S$ can be answered optimally.

For this purpose, we introduce an {\em adapter} for each final-level PSS query instance.
And the space consumption of each such adapter is bounded by~$O(1)$ words.
As there are at most~$O(n_0)$ final-level PSS query instances, the total space consumption of the adapters is bounded by~$O(n_0)$.
Recall that~$n_0$ is the size of the item set~$S$ when the HALT structure is constructed, and~$n$ is the {\em current} size of~$S$. 
Without loss of generality, we assume that $n_0/2 \leq n \leq 2n_0$. 
If this is not the case, we can rebuild the HALT structure on the current~$S$ to make this assumption hold.
More details about global rebuilding and de-amortization are discussed in Section~\ref{sec:updates}. 

Recall that in the three-level hierarchy, at Level-1, the next-level item set $Y_j$ of each group $G_S(j)$ in BG-Str($S$) contains at most $\log_2 n_0$ items; and at Level-2, the next-level item set $Z_k$ of each group $G_{Y_j}(k)$ in BG-Str($Y_j$) has size at most $\log_2 \log_2 n_0$.
Furthermore, at Level-3, each item $v_i$ in the item set $V_k$ in the final-level instance corresponds to a non-empty bucket,~$B_{Z_k}(i)$, of items in~$Z_k$, where each item $z \in Z_k$ corresponds to a bucket in group~$G_{Y_j}(k)$.
Thus, we have $|V_k| \leq |Z_k| \leq m$, where $m  = \log_2 \log_2 n_0$.
And the weight of each item $v_i \in V_k$ is $w(v_i) = 2^{i+1} \cdot |B_{Z_k}(i)| \leq 2^{i+1} \cdot m$.

\vspace{2mm}
\noindent
{\bf A Simple but Space-Inefficient Adapter Implementation.}
Define $m = \log_2 \log_2 n_0$.
Consider the final-level instance on $Z_k$ as aforementioned. 
A simple implementation of an adapter is to maintain an array $\mathcal{A}$ of length $L$, where $L$ is number of possible bucket indices of the items in $Z_k$, and for $i \in [0, L-1]$, $\mathcal{A}[i]$ stores the size of (possibly empty) bucket $B_{Z_k}(i)$.
Moreover, as $|Z_k| \leq m$, each bucket size must be an integer in $[0, m]$.

\vspace{2mm}
\noindent
\underline{\em The Query Algorithm for a Final-Level Instance.} 
Given a PSS query with parameters $(0, W_S(\a, \b))$ on $V_k$, the following query algorithm aims to return a PSS result $T \subseteq Z_k$ with parameters $(0, W_S(\a, \b))$.

\vspace{2mm}
\noindent
{\bf QueryFinalLevel(} BG-Str($Z_k$), $(0, W_S(\a, \b))$ {\bf )}.
\begin{itemize}[leftmargin = *]
\item compute the {\em largest} possible index $i_1$ of bucket $B_{Z_k}(i_1)$ such that $\frac{2^{i_1 + 1}}{W_S(\a, \b)} \leq \frac{2}{m^2}$; 
\item compute the {\em smallest} possible index $i_2$ of bucket $B_{Z_k}(i_2)$ such that $\frac{2^{i_2}}{W_S(\a, \b)} \geq 1$;

\item $T_1 \leftarrow$ QueryInsignificant(BG-Str($Z_k$), $(0, W_S(\a,\b)
), i_1$) \text{//} Algorithm~\ref{algo:insig};

\item $T_2 \leftarrow$ QueryCertain(BG-Str($Z_k$), $(0, W_S(\a, \b))$, $i_2$) \text{//} Algorithm~\ref{algo:cert};

\item 
process all the remaining buckets $B_{Z_k}(i)$ with index $i \in [i_1 + 1, i_2 -1]$:
\begin{itemize}
\item collect all the bucket sizes in $\mathcal{A}[i_1 + 1, i_2 -1]$ as an input configuration $\vec{c}$ for the 4S problem;
specifically, recall that the $j^\text{th}$ item in the 4S problem is selected with probability $p_j = \min\left\{1, \frac{2^{j+1} \cdot \vec{c}[j]}{m^2}\right\}$, where $\vec{c}[j] = |B_{Z_k}(j + i_1)|$;

\item use the lookup table $\mathcal{T}$ to obtain a subset sampling result $r_3$ for the 4S problem with input configuration $\vec{c}$;
note that $r_3$ is a bit string of length $(i_2 - i_1 -1)$, where the $j^\text{th}$ bit in $r_3$ being $1$ indicates the bucket with index $(j + i_1)$ is selected;
\item initialize an empty set $T_3'$;
\item for each index $j \in \{1, 2, \ldots, i_2 - i_1 -1\}$ such that $r_3[j] = 1$, add the item $v \in V_k$ corresponding bucket $B_{Z_k}(j + i_1)$ to $T_3'$ with probability $\frac{\min\left\{1, \frac{w(v)}{W_S(\a,\b)}\right\}}{p_j}$;
\item $T_3 \leftarrow$ ExtractItems(BG-Str($Z_k$), $T_3'$, $(0, W_S(\a, \b)$) \text{//} Algorithm~\ref{algo:extract}
\end{itemize}
\item return $T = T_1 \cup T_2 \cup T_3$; 
\end{itemize}
\vspace{-1mm}
\begin{lemma}
The configuration $\vec{c}$ is a legal 4S problem input defined for the lookup table $\mathcal{T}$.
\end{lemma}
\begin{proof} 

From the above query algorithm, $\vec{c}$ contains 
exactly $i_2 - 1 - (i_1 + 1) + 1 = i_2 - i_1 - 1$ integers. 
Observe that all the buckets with indices in the range $[i_1 + 1, i_2 -1]$
must correspond to a power-of-two probability 
in the range $[\frac{1}{m^2}, 1]$. 
Therefore, there are at most $\lceil 2 \log_2 m\rceil \in O(\log m) = O(\log \log \log n_0)$ possible buckets.
And hence, $\vec{c}$ contains exactly $K \in O(\log \log \log n_0)$ integers.
Moreover, each integer is indeed a bucket size which is in the range $[0, m]$.
Therefore, $\vec{c}$ is a legal input configuration to the 4S problem.
\end{proof}
\vspace{-1mm}
\begin{lemma}[Correctness]
$T$ is a correct PSS result on $Z_k$ with parameters $(0, W_S(\a, \b))$. 
\end{lemma}
\begin{proof}
The correctness of the sampling result $T_1 \cup T_2$ follows immediately from the correctness of our algorithms for the insignificant instance and the certain instance.
By the correctness of the lookup table and the rejection sampling approach,
each item $v\in V_k$ corresponding to a non-empty bucket $B_{Z_k}(i)$ with index in $\mathcal{A}[i_1 + 1, i_2 -1]$ is sampled with probability of $\min\{1, \frac{w(v)}{W_S(\a, \b)}\}$ independently.
Thus, by the correctness of Algorithm~\ref{algo:extract} which extracts the items from a given candidate bucket list, the sampling result $T_3$ is correct, and hence, $T$ is correct.
\end{proof}
\vspace{-1mm}
\begin{lemma}[Query Complexity]\label{lmm:adapter-query}
The overall running time of the above final-level query algorithm is bounded by $O(1 + \mu_{Z_k})$ in expectation, where $\mu_{Z_k}$ is the expected size of PSS query result on $Z_k$ with parameters $(0, W_S(\a, \b))$.
\end{lemma}
\begin{proof}
First, the expected query running time for obtaining items in $T_1 \cup T_2$ follows from the complexity of the corresponding algorithms.

Next, we analyse the running time cost of obtaining~$T_3$.
By the definition of~$i_1$ and~$i_2$, the length of the subarray $\mathcal{A}[i_1 + 1, i_2 -1]$ is bounded $i_2 - i_1 -1 \in O(\log m)$.
Each bucket size is an integer at most~$m$; thus, each bucket size can be encoded in~$O(\log m)$ bits. 
Therefore, $\mathcal{A}[i_1 + 1, i_2 -1]$ takes $O(\log^2 m) = O((\log \log \log n_0)^2) \subset O(d)$ bits, equivalently,~$O(1)$ words.  
As a result, the input configuration~$\vec{c}$ can be constructed in~$O(1)$ time.
Moreover, the lookup table returns a sampling result~$r_3$ 
for the 4S problem on~$\vec{c}$ in~$O(1)$ time. 
In the Word RAM model, the next non-zero bit~$j$ in~$r_3$ can be computed in~$O(1)$ time.
Thus, the overall expected running time of rejection sampling 
is bounded by~$O(J)$, where~$J$ is the number of non-zero bits in~$r_3$.
To complete the proof of this lemma, 
it remains to show that  each potential item $v$ suggested by a non-zero bit in $r_3$
has at least~$\Omega(1)$ probability to be accepted in~$T_3$. 
If this is the case, then~$O(J)$ as well as the overall expected running time of our query algorithm, is bounded by $O(1 + \mu_{V_k})$. 

Consider a non-zero bit~$r_3[j] = 1$ at index~$j$;
let $v \in V_k$ be the potential item corresponding to the bucket $B_{Z_k}(j+ i_1)$ which is sampled by the lookup table with probability $p_j = \min\{1, \frac{2^{j+1}\cdot |B_{Z_k}(j+i_1)|}{m^2}\}$.
Denote 
$p'_j = \min\{1, \frac{w(v)}{W_S(\a, \b)}\}$, where $w(v) = 2^{j+i_1 + 1} \cdot |B_{Z_k}(j + i_1)|$. 
Then $v$ is accepted with probability $\frac{p'_j}{p_j} \geq \frac{2^{i_1}\cdot m^2}{W_S(\a, \b)} \geq \frac{1}{2}$, because, by the definition of $i_1$, we have $\frac{2^{i_1 + 1} m^2}{W_S(\a, \b)} \geq 1$.
Moreover, according to the previous running time analysis on Algorithm~\ref{algo:extract}, it takes $O(1 + \mu_{Z_k})$ expected time to extract the items from the bucket candidate list $T_3'$. 
Therefore, Lemma~\ref{lmm:adapter-query} follows.
\end{proof}

\vspace{-1mm}
The above lemmas show that the simple implementation with an array $\mathcal{A}$ for the adapter is correct and has the desired query complexity. 
However, its space consumption is too large. 
This is because the length,~$L$, of~$\mathcal{A}$ is the number of all possible bucket indices on item set,~$Z_k$. 
The weight of item $z \in Z_k$ can range from $[1, n_{\max}\cdot w_{\max}]$, and hence, $L = \log_2 (n_{\max}\cdot w_{\max}) = d$.
Moreover, the bucket size stored in each cell of~$\mathcal{A}$ is an integer in~$[0, m]$, which in turn takes~$\lceil \log_2 (m+1)\rceil\in O(\log m)$ bits to encode.
Therefore, the total space consumption of~$\mathcal{A}$ becomes $O(d \log m) $ bits, equivalently, $O(\log m) \subseteq O(\log \log \log n_0)$ words.
As aforementioned, there can be up to~$O(n_0)$ final-level instances, each which needs an adapter. Thus, the overall space consumption of all the adapters would be $O(n_0 \log \log \log n_0)$ words, exceeding our promised space bound~$O(n_0)$.

\vspace{2mm}
\noindent
{\bf A Compact Representation of~$\mathcal{A}$.}
While the total number,~$L$, of possible bucket indices 
can be as large as~$d$, 
our crucial observation is that, indeed, only a ``small'' consecutive index sub-range $[l_1, l_2] \subseteq [0, L-1]$ 
can possibly have non-empty buckets of the items in~$Z_k$.
As a result, it suffices to store the sizes of the buckets with indices in $[l_1, l_2]$, that is, just to keep a {\em sub-array} $\mathcal{A}[l_1, l_2]$ of the conceptual array $\mathcal{A}$, along with the value of $l_1$ (to record the starting index of the sub-range).
Hence, the space consumption can be reduced to $O((l_2 - l_1 + 1) \cdot \log m + d)$ bits, where the $d$ bits are to store the value of $l_1$.
This sub-array $\mathcal{A}[l_1, l_2]$ along with the value of $l_1$ together is called a {\em compact representation} of $\mathcal{A}$.

When answering a PSS query with the previous query algorithm with a compact representation of $\mathcal{A}$,
one can simply treat every bucket met during the query with index outside the range $[l_1, l_2]$ as empty, i.e., of size $0$.
The correctness and query complexity of the query algorithm still holds.  
Next, we show the length of $\mathcal{A}[l_1, l_2]$ is bounded by $O(\log \log n_0)$.

\begin{lemma}
For every final-level PSS instance on $Z_k$, the space consumption of the compact representation of the adapter $\mathcal{A}$ is bounded by $O(1)$ words.  
\end{lemma}
 
\begin{proof}
Recall that each item $v_i\in V_k$ corresponds to a bucket $B_{Z_k}(i)$ of items in $Z_k$ at bucket index $i$. 
Our goal is to show that those non-empty buckets $B_{Z_k}(i)$ must have bucket index $i$ in range $[l_1, l_2]$, for some~$l_1$ and~$l_2$. 

To prove this, 
we show that the {\em minimum possible} weight and the {\em maximum possible} weight of the items in~$Z_k$ would not differ by too much. 
To see this, observe that the item set~$Z_k$ is constructed from the buckets in the group $G_{Y_j}(k)$, where $G_{Y_j}(k)$ contains buckets $B_{Y_j}(i)$ with bucket indices from $k_1 = k \cdot \log_2 |Y_j|$ to $ k_2 = (k + 1) \cdot \log_2 |Y_j|$.
Therefore, the minimum possible weight of a bucket in $G_{Y_j}(k)$ is $2^{k_1+1}$, while the maximum possible weight of a bucket in $G_{Y_j}(k)$ is $2^{k_2 + 1} \cdot |Y_j|$. 
And each of the buckets corresponds to an item $z_i \in Z_k$ with the same weight. 
As a result, the {\em possible weight ratio} of the items in $Z_k$, i.e.~the ratio between the maximum and minimum possible weights,  can be computed as 
$$ \frac{2^{k_2 + 1} \cdot |Y_j|}{2^{k_1 + 1}} = 2^{k_2 - k_1} \cdot |Y_j| = 2^{\log_2 |Y_j|} \cdot |Y_j| = |Y_j|^2 \leq (\log_2 n_0)^2\,,$$
where $|Y_j| \leq \log_2 n_0$.
Thus, there can be at most $\log_2 \left( (\log_2 n_0)^2\right)  = 2 \log_2 \log_2 n_0 \in O(\log \log n_0)$ possible buckets on the items in $Z_k$.
Moreover, 
since the smallest possible bucket weight of $B_{Z_k}(i)$ is $2^{k_1 + 1}$,
the smallest possible bucket index for $B_{Z_k}(i)$ is $l_1 = k_1 + 1 = k\cdot \log_2 |Y_j| + 1$.
And hence, $l_2 \leq l_1 + 2 \log_2 \log_2 n_0$. 

By the previous discussion, a compact representation of~$\mathcal{A}$ consists of a sub-array $\mathcal{A}[l_1, l_2]$, where each cell stores an integer in range~$[0, m]$, and an integer value~$l_1$.
The total space consumption is thus bounded by $(l_2 - l_1 + 1) \cdot \lceil \log_2 (m+1)\rceil + d \in O(\log \log n_0 \cdot \log \log \log n_0 + d)$ bits which is bounded by $O(1)$ words.
\end{proof}

\subsection{Handling Updates}

% It can be verified that the HALT structure is easy to update by just maintaining the sorted linked lists of the buckets and the groups in $O(1)$ time.
% To handle the case when $n < n_0/ 2$ or $n > 2n_0$, a global rebuilding strategy is adopted. This, however, makes the update cost become {\em amortized}.
% By applying the same de-amortization technique for the standard dynamic arrays, the update cost can be de-amortized back to $O(1)$ worst-case.
% We move the detailed discussion to Appendix~\ref{sec:updates}.

% \begin{toappendix}
% \input{update-algo}
% \end{toappendix}

Next, we discuss how to update our HALT structure on the item set~$S$. 
Specifically, each update on~$S$ is either an insertion of a new item~$x$ with weight~$w(x)$ or a deletion of an existing item~$x$ from~$S$. 
For simplicity, we assume that $\frac{1}{2} n_0 \leq n \leq 2n_0$, where~$n$ is the current size of~$|S|$ after the update.

The HALT structure on~$S$ can be maintained as follows.
First, observe that the lookup table in the HALT structure on~$S$ does not need to be updated.
Second, the three-level sampling hierarchy is basically a hierarchy of sorted lists of buckets and groups with length at most~$O(d)$. Each update operation (either insertion or deletion) on these sorted lists can be easily done in~$O(1)$ time in the Word RAM Model. 
The only issue to be careful about is that
each update operation at the previous level in the hierarchy could trigger a {\em weight update} of a next-level item, which would require two update operations to achieve: delete the item first and then insert the item again with the new weight. 
As the hierarchy has only three levels, therefore, there can be in total at most $O(1)$ update operations to the hierarchy. 
Moreover, there is also only $O(1)$ final-level instances would be affected by an update to $S$. 
Their adapters can be updated simply by updating the corresponding bucket sizes in the sub-array $\mathcal{A}[l_1, l_2]$.
Therefore, the overall update time is bounded by $O(1)$ worst-case time.

To handle the case when $n < \frac{1}{2} n_0$ or $n > 2 n_0$,
we adopt the standard global rebuilding trick to re-construct the entire HALT structure on the current item set~$S$ and update $n_0 \leftarrow n$. 
Recall that the HALT structure on~$S$ can be constructed in~$O(n_0)$ time with~$O(n_0)$ space consumption.
As there must be at least~$\Omega(n_0)$ updates to trigger the reconstruction, 
the rebuilding cost can be charged to those updates making each update cost become~$O(1)$ amortized.
Moreover, the space consumption is bounded by~$O(n)$ all the times.

Interestingly, the~$O(1)$ amortized update time can be easily ``de-amortized'' by applying the same technique for the de-amortization for dynamic arrays, just increasing the space consumption by a constant factor. 
Thus, we have the following theorem.
\begin{theorem}
The HALT structure on item set~$S$ can be maintained in~$O(1)$ worst-case time for each update to~$S$. Space consumption is bounded by~$O(n)$ at all times, where~$n$ is the current size of~$S$. 
\end{theorem}
\section{A Hardness Result on the DPSS Problem with Float Item Weights}

We now prove Theorem~\ref{thm:hardness}. Let {\em DPSS-ALG} be a deletion-only DPSS  algorithm for float item weights with pre-processing time $\Pt(N)$, query time $\Qt(N)$, and deletion time $\Dt(N)$. As a black box, we apply the algorithm to execute Integer Sorting.

Consider a set of $N$ integers $I = \{a_1, \ldots, a_N\}$, each of which is represented by one word of $d$ bits. We can assume w.l.o.g.~that the integers are distinct\footnote{To achieve distinctness, to each integer, we can append a unique ID with a word in its binary representation.}.
The integers in $I$ can be sorted in descending order by the following algorithm:
\begin{itemize}[leftmargin = *]
\item for each integer $a_i \in I$, create an item $x_i$ with weight $w(x_i) = 2^{a_i}$, represented by a float number;
\item initialize $S$ to be the set of all these $N$ items; 
\item initialize an empty linked list,~$R$, of the integers in~$I$, which is maintained to be sorted, in descending order, by the Insertion Sort~\cite{gill2019comparative} algorithm; 
\item initialize {\em DPSS-ALG} on $S$;
\item while $S$ is not empty, perform the following:
\begin{itemize}
\item {\em repeatedly} invoke {\em DPSS-ALG} on $S$ to perform a PSS query with parameters $(1, 0)$ until the sampling result $T \not\eq \emptyset$;
\item let $x^*$ be the item in $T$ with the {\em largest} weight and let $w(x^*)=2^{a^*}$;
\item invoke {\em DPSS-ALG} to delete $x^*$ from $S$;
\item invoke Insertion Sort to insert the weight exponent, $a^*$, of $x^*$ to $R$ (in descending order); 
\end{itemize}
\item return $R$ as the sorted list of all the integers in $I$ (in descending order);
\end{itemize}

\vspace{2mm}
\noindent
{\bf Correctness.}
When the above algorithm terminates, $R$ contains all the $N$ integers in $I$.
As $R$ is maintained by the Insertion Sort algorithm, $R$ is a correct sorted list of $I$. 

\vspace{2mm}
\noindent
{\bf Running Time Analysis.}
Next we analyse the overall expected running time 
of the above integer sorting algorithm. 
First, the construction cost of the item set $S$ is bounded by $O(N)$ and the pre-processing cost of {\em DPSS-ALG} is $t_{p}(N)$. Next, we focus on the running time in each iteration in the while-loop.
Denote the item set at the start of the $i^{\text{th}}$ iteration by $S_i$, where $i \in \{1, \ldots, N\}$ and $S_1 = S$.

\begin{lemma}\label{lmm:trial-num}
In the $i^\text{th}$ iteration, in expectation, at most two PSS queries with parameters $(1,0)$ on $S_i$ are invoked to obtain a non-empty subset sampling result,~$T$. 
\end{lemma}
\begin{proof}
To prove this lemma, it suffices to show that the {\em largest item} (i.e., with the largest item weight) in $S_i$ will be selected in the sampling result $T$ with probability at least $\frac{1}{2}$. 
As a result, with at most $2$ trials, in expectation, the largest item would be sampled in $T$, and hence, $T$ is non-empty. 

To see this, first, observe that with the query parameters $(1,0)$, we have $W_{S_i}(1,0) = \sum_{x\in S_i} w(x)$. 
Second, by the construction of $S$, the weights of all the items are power-of-two numbers with {\em distinct exponents}. 
Therefore, it can be verified that the largest item in $S_i$ must have weight $\geq \frac{1}{2} W_{S_i}(1,0)$, and hence, the corresponding item has probability at least~$\frac{1}{2}$ to be sampled.
\end{proof}

\begin{lemma}\label{lmm:expected-size}
For a PSS query with parameters $(1, 0)$ on $S_i$, the expected result size $\mu_{S_i}(1,0) = 1$. 
\end{lemma}
\begin{proof}
It follows from the fact that $W_{S_i}(1,0) = \sum_{x\in S_i} w(x)$ and $\mu_{S_i}(1,0) = \sum_{x\in S_i} \frac{w(x)}{W_{S_i}(1,0)}$. 
\end{proof}

By Lemmas~\ref{lmm:trial-num} and~\ref{lmm:expected-size}, 
in the $i^{\text{th}}$ iteration, the running time for the steps 
of sampling a non-empty subset~$T$ and finding the largest item~$x^*$ in~$T$ is bounded by~$O(t_q(N))$ in expectation.
Moreover, {\em DPSS-ALG}  deletes $x^*$ from $S_i$ in $t_{del}(N)$ worst-case time.  
As a result, summing over all $N$ iterations, the overall expected cost for these steps is bounded by $O(N\cdot (t_q(N) + t_{del}(N)))$. 

Next, we analyse the running time for the Insertion Sort.

\begin{lemma}
    \label{lmm:insertion-sort}
    The expected running time for the Insertion Sort is $O(N)$.
\end{lemma}
\begin{proof}
The order of the integers in~$I$ is identical to the order of their corresponding weights.
To ease the discussion, we treat the sorted linked list~$R$ of the integers in~$I$ as the sorted linked list of their corresponding items.
It is known that the overall cost for Insertion Sort is bounded by
$O(N + \sum_{i=1}^N \#\text{Swap}_i)$, 
where $\#\text{Swap}_i$ is the number of {\em swaps} for the item $x^*$ which is inserted to the sorted linked list $R$ from the back in the $i^\text{th}$ iteration.
Equivalently, $\#\text{Swap}_i$ is the 
number of items that have weights {\em smaller than} $w(x^*)$ and were inserted to $R$ {\em before} the $i^\text{th}$ iteration. 

Define $\#\text{FutureSwap}_i$ as 
the number of  items that have weights {\em greater than} $w(x^*)$ and will be inserted to $R$ {\em after} the $i^\text{th}$ iteration.
Observe that $\sum_{i=1}^N \#\text{Swap}_i = \sum_{i=1}^N \#\text{FutureSwap}_i$.
Furthermore, define the {\em rank} of an item $x \in S_i$, denoted by $\text{rank}_{S_i}(x)$, as the number of items that have weight greater than~$w(x)$ in~$S_i$.
Then, we have 
$\#\text{FutureSwap}_i = \text{rank}_{S_i}(x^*)$, where $x^*$ is the item inserted to~$R$ in the $i^\text{th}$ iteration. 
This follows from the definition of $\#\text{FutureSwap}_i$ and the fact that $S_i \setminus \{x^*\}$ is exactly the set of items inserted to $R$ in future iterations. 
As a result, it suffices to bound the expectation of $\text{rank}_{S_i}(x^*)$.
\begin{claim}
In the $i^\text{th}$ iteration, 
consider the largest item $x^*$ in the sample set $T \subseteq S_i$. We have that 
 $\mathbb{E}[\text{rank}_{S_i}(x^*)] \in O(1)$.
\end{claim}
\begin{nestedproof}
As discussed earlier, the item weights are power-of-two numbers with distinct exponents.
We have the following observations.
First, the largest item $y_1 \in S_i$ would be sampled in the PSS query on $S_i$ with parameter $(1,0)$ with probability at least $\frac{1}{2}$ and at most $1$. 
Thus, 
the probability of the $j^\text{th}$ largest item $y_j$ in $S_i$ to be sampled is at most $\frac{1}{2^{j-1}}$, for all $j \in \{1, \ldots, |S_i|\}$.
Second, the $j^\text{th}$ largest item $y_j$ in $S_i$ has  $\text{rank}_{S_i}(y_j) = j - 1$. 

Moreover, the probability that the $j^\text{th}$ largest item $y_j$ in $S_i$ is the largest item in the sample set $T$ is at most the probability of $y_j$ being sampled conditioned on $T$ being non-empty.
Therefore, the expected rank of $x^*$ satisfies: 
\begin{align*}
\mathbb{E}[\text{rank}_{S_i}(x^*)] 
&\leq \sum_{j=1}^{|S_i|} \text{rank}_{S_i}(y_j) \cdot \text{Pr[ $y_j$ being sampled $\mid$ $T$ is non-empty]} \\
&= \sum_{j=1}^{|S_i|} \text{rank}_{S_i}(y_j) \cdot \frac{\text{Pr[ $y_j$ being sampled $\land$ $T$ is non-empty]}}{\text{Pr[ $T$ is non-empty ]}} \\
&\leq \sum_{j=1}^{|S_i|} (j-1) \cdot \frac{2}{2^{j-1}} \in O(1)\,,
\end{align*}
where the last inequality is by the fact that $T$ is non-empty with probability at least $\frac{1}{2}$.
\end{nestedproof}

Therefore, the running time of the Insertion Sort over all the $N$ iterations is bounded by $O(N)$ in expectation.
This completes the proof of Lemma~\ref{lmm:insertion-sort}.
\end{proof}
Putting together all of the above analyses, the total expected running time of the algorithm is $O(N)+t_p(N)+O(N \cdot (t_q(N)+t_{del}(N)))+O(N)$, which is $t_p(N) + O(N \cdot (t_q(N) + t_{del}(N)))$.

Theorem~\ref{thm:hardness} thus follows.

\section{Conclusion}

In this paper, we study the Dynamic Parameterized Subset Sampling (DPSS) problem in the Word RAM model. 
Our first main result is an optimal DPSS algorithm. 
We propose our novel data structure HALT that achieves the optimal bounds.
Our second main result showing that if the item weights in the DPSS problem are allowed to be float numbers, then 
an deletion-only optimal DPSS algorithm can be used to solve Integer Sorting over the whole range of bit-length in $O(N)$ expected time, where $N$ is the number of integers to sort.
Finally, 
our third main result is a new efficient method for sampling variates from the Truncated Geometric distribution in $O(1)$ expected time in the Word RAM model.
We expect considerable application of this technique for efficient random variate realization.

\section*{Acknowledgements}

In this work, Junhao Gan is supported in part by Australian Research Council (ARC) Discovery Early Career Research Award (DECRA) DE190101118.
Seeun William Umboh is partially supported by 
ARC Training Centre in Optimisation Technologies, Integrated Methodologies, and Applications (OPTIMA).
Hanzhi Wang is supported by the VILLUM Foundation grant 54451.

\newpage

\bibliographystyle{ACM-Reference-Format}
\bibliography{main}


\begin{thebibliography}{35}

%%% the Web address.

\ifx \showCODEN    \undefined \def \showCODEN     #1{\unskip}     \fi
\ifx \showDOI      \undefined \def \showDOI       #1{#1}\fi
\ifx \showISBNx    \undefined \def \showISBNx     #1{\unskip}     \fi
\ifx \showISBNxiii \undefined \def \showISBNxiii  #1{\unskip}     \fi
\ifx \showISSN     \undefined \def \showISSN      #1{\unskip}     \fi
\ifx \showLCCN     \undefined \def \showLCCN      #1{\unskip}     \fi
\ifx \shownote     \undefined \def \shownote      #1{#1}          \fi
\ifx \showarticletitle \undefined \def \showarticletitle #1{#1}   \fi
\ifx \showURL      \undefined \def \showURL       {\relax}        \fi
% invisible to TeX
\providecommand\bibfield[2]{#2}
\providecommand\bibinfo[2]{#2}
\providecommand\natexlab[1]{#1}
\providecommand\showeprint[2][]{arXiv:#2}

\bibitem[Andersen et~al\mbox{.}(2006)]%
        {andersen2006FOCS}
\bibfield{author}{\bibinfo{person}{Reid Andersen}, \bibinfo{person}{Fan Chung}, {and} \bibinfo{person}{Kevin Lang}.} \bibinfo{year}{2006}\natexlab{}.
\newblock \showarticletitle{Local graph partitioning using pagerank vectors}. In \bibinfo{booktitle}{\emph{2006 47th Annual IEEE Symposium on Foundations of Computer Science (FOCS'06)}}. IEEE, \bibinfo{pages}{475--486}.
\newblock


\bibitem[Baldwin et~al\mbox{.}(2017)]%
        {baldwin2017contagion}
\bibfield{author}{\bibinfo{person}{Adrian Baldwin}, \bibinfo{person}{Iffat Gheyas}, \bibinfo{person}{Christos Ioannidis}, \bibinfo{person}{David Pym}, {and} \bibinfo{person}{Julian Williams}.} \bibinfo{year}{2017}\natexlab{}.
\newblock \showarticletitle{Contagion in cyber security attacks}.
\newblock \bibinfo{journal}{\emph{Journal of the Operational Research Society}} \bibinfo{volume}{68}, \bibinfo{number}{7} (\bibinfo{year}{2017}), \bibinfo{pages}{780--791}.
\newblock


\bibitem[Belazzougui et~al\mbox{.}(2014)]%
        {belazzougui2014expected}
\bibfield{author}{\bibinfo{person}{Djamal Belazzougui}, \bibinfo{person}{Gerth~St{\o}lting Brodal}, {and} \bibinfo{person}{Jesper~Sindahl Nielsen}.} \bibinfo{year}{2014}\natexlab{}.
\newblock \showarticletitle{Expected Linear Time Sorting for Word Size $\Omega(\log^2 n \log \log n)$}. In \bibinfo{booktitle}{\emph{Scandinavian Workshop on Algorithm Theory}}. Springer, \bibinfo{pages}{26--37}.
\newblock


\bibitem[Benson et~al\mbox{.}(2016)]%
        {benson2016Science}
\bibfield{author}{\bibinfo{person}{Austin~R Benson}, \bibinfo{person}{David~F Gleich}, {and} \bibinfo{person}{Jure Leskovec}.} \bibinfo{year}{2016}\natexlab{}.
\newblock \showarticletitle{Higher-order organization of complex networks}.
\newblock \bibinfo{journal}{\emph{Science}} \bibinfo{volume}{353}, \bibinfo{number}{6295} (\bibinfo{year}{2016}), \bibinfo{pages}{163--166}.
\newblock


\bibitem[Bhattacharya et~al\mbox{.}(2023)]%
        {bhattacharya2023graphmatching}
\bibfield{author}{\bibinfo{person}{Sayan Bhattacharya}, \bibinfo{person}{Peter Kiss}, \bibinfo{person}{Aaron Sidford}, {and} \bibinfo{person}{David Wajc}.} \bibinfo{year}{2023}\natexlab{}.
\newblock \showarticletitle{Near-Optimal Dynamic Rounding of Fractional Matchings in Bipartite Graphs}.
\newblock \bibinfo{journal}{\emph{arXiv preprint arXiv:2306.11828}} (\bibinfo{year}{2023}).
\newblock


\bibitem[Brent and Zimmermann(2010)]%
        {brent2010modern}
\bibfield{author}{\bibinfo{person}{Richard~P Brent} {and} \bibinfo{person}{Paul Zimmermann}.} \bibinfo{year}{2010}\natexlab{}.
\newblock \bibinfo{booktitle}{\emph{Modern computer arithmetic}}. Vol.~\bibinfo{volume}{18}.
\newblock \bibinfo{publisher}{Cambridge University Press}.
\newblock


\bibitem[Bringmann and Friedrich(2013)]%
        {bringmann2013exact}
\bibfield{author}{\bibinfo{person}{Karl Bringmann} {and} \bibinfo{person}{Tobias Friedrich}.} \bibinfo{year}{2013}\natexlab{}.
\newblock \showarticletitle{Exact and efficient generation of geometric random variates and random graphs}. In \bibinfo{booktitle}{\emph{International Colloquium on Automata, Languages, and Programming}}. Springer, \bibinfo{pages}{267--278}.
\newblock


\bibitem[De~Berg(2000)]%
        {de2000computational}
\bibfield{author}{\bibinfo{person}{Mark De~Berg}.} \bibinfo{year}{2000}\natexlab{}.
\newblock \bibinfo{booktitle}{\emph{Computational geometry: algorithms and applications}}.
\newblock \bibinfo{publisher}{Springer Science \& Business Media}.
\newblock


\bibitem[Esfandiari et~al\mbox{.}(2021)]%
        {esfandiari2021almost}
\bibfield{author}{\bibinfo{person}{Hossein Esfandiari}, \bibinfo{person}{Vahab Mirrokni}, {and} \bibinfo{person}{Peilin Zhong}.} \bibinfo{year}{2021}\natexlab{}.
\newblock \showarticletitle{Almost linear time density level set estimation via dbscan}. In \bibinfo{booktitle}{\emph{Proceedings of the AAAI Conference on Artificial Intelligence}}, Vol.~\bibinfo{volume}{35}. \bibinfo{pages}{7349--7357}.
\newblock


\bibitem[Flajolet and Saheb(1986)]%
        {flajolet1986complexity}
\bibfield{author}{\bibinfo{person}{Philippe Flajolet} {and} \bibinfo{person}{Nasser Saheb}.} \bibinfo{year}{1986}\natexlab{}.
\newblock \showarticletitle{The complexity of generating an exponentially distributed variate}.
\newblock \bibinfo{journal}{\emph{Journal of Algorithms}} \bibinfo{volume}{7}, \bibinfo{number}{4} (\bibinfo{year}{1986}), \bibinfo{pages}{463--488}.
\newblock


\bibitem[Fleming and Cannon(2018)]%
        {fleming2018stochastic}
\bibfield{author}{\bibinfo{person}{James Fleming} {and} \bibinfo{person}{Mark Cannon}.} \bibinfo{year}{2018}\natexlab{}.
\newblock \showarticletitle{Stochastic MPC for additive and multiplicative uncertainty using sample approximations}.
\newblock \bibinfo{journal}{\emph{IEEE Trans. Automat. Control}} \bibinfo{volume}{64}, \bibinfo{number}{9} (\bibinfo{year}{2018}), \bibinfo{pages}{3883--3888}.
\newblock


\bibitem[Fredman and Willard(1993)]%
        {fredman1993surpassing}
\bibfield{author}{\bibinfo{person}{Michael~L Fredman} {and} \bibinfo{person}{Dan~E Willard}.} \bibinfo{year}{1993}\natexlab{}.
\newblock \showarticletitle{Surpassing the information theoretic bound with fusion trees}.
\newblock \bibinfo{journal}{\emph{Journal of computer and system sciences}} \bibinfo{volume}{47}, \bibinfo{number}{3} (\bibinfo{year}{1993}), \bibinfo{pages}{424--436}.
\newblock


\bibitem[Germann et~al\mbox{.}(2006)]%
        {germann2006mitigation}
\bibfield{author}{\bibinfo{person}{Timothy~C Germann}, \bibinfo{person}{Kai Kadau}, \bibinfo{person}{Ira~M Longini~Jr}, {and} \bibinfo{person}{Catherine~A Macken}.} \bibinfo{year}{2006}\natexlab{}.
\newblock \showarticletitle{Mitigation strategies for pandemic influenza in the United States}.
\newblock \bibinfo{journal}{\emph{Proceedings of the National Academy of Sciences}} \bibinfo{volume}{103}, \bibinfo{number}{15} (\bibinfo{year}{2006}), \bibinfo{pages}{5935--5940}.
\newblock


\bibitem[Gill et~al\mbox{.}(2019)]%
        {gill2019comparative}
\bibfield{author}{\bibinfo{person}{Sandeep~Kaur Gill}, \bibinfo{person}{Virendra~Pal Singh}, \bibinfo{person}{Pankaj Sharma}, {and} \bibinfo{person}{Durgesh Kumar}.} \bibinfo{year}{2019}\natexlab{}.
\newblock \showarticletitle{A comparative study of various sorting algorithms}.
\newblock \bibinfo{journal}{\emph{International Journal of Advanced Studies of Scientific Research}} \bibinfo{volume}{4}, \bibinfo{number}{1} (\bibinfo{year}{2019}).
\newblock


\bibitem[Guo et~al\mbox{.}(2020)]%
        {subsim_influence_guo}
\bibfield{author}{\bibinfo{person}{Qintian Guo}, \bibinfo{person}{Sibo Wang}, \bibinfo{person}{Zhewei Wei}, {and} \bibinfo{person}{Ming Chen}.} \bibinfo{year}{2020}\natexlab{}.
\newblock \showarticletitle{Influence Maximization Revisited: Efficient Reverse Reachable Set Generation with Bound Tightened}. In \bibinfo{booktitle}{\emph{{SIGMOD} Conference}}. \bibinfo{publisher}{{ACM}}, \bibinfo{pages}{2167--2181}.
\newblock


\bibitem[Guo et~al\mbox{.}(2022)]%
        {subsim_tods}
\bibfield{author}{\bibinfo{person}{Qintian Guo}, \bibinfo{person}{Sibo Wang}, \bibinfo{person}{Zhewei Wei}, \bibinfo{person}{Wenqing Lin}, {and} \bibinfo{person}{Jing Tang}.} \bibinfo{year}{2022}\natexlab{}.
\newblock \showarticletitle{Influence maximization revisited: efficient sampling with bound tightened}.
\newblock \bibinfo{journal}{\emph{ACM Transactions on Database Systems (TODS)}} \bibinfo{volume}{47}, \bibinfo{number}{3} (\bibinfo{year}{2022}), \bibinfo{pages}{1--45}.
\newblock


\bibitem[Hajar et~al\mbox{.}(2019)]%
        {hajar2019discrete}
\bibfield{author}{\bibinfo{person}{Mayssa Hajar}, \bibinfo{person}{Mohamad El~Badaoui}, \bibinfo{person}{Amani Raad}, {and} \bibinfo{person}{Frederic Bonnardot}.} \bibinfo{year}{2019}\natexlab{}.
\newblock \showarticletitle{Discrete random sampling: Theory and practice in machine monitoring}.
\newblock \bibinfo{journal}{\emph{Mechanical Systems and Signal Processing}}  \bibinfo{volume}{123} (\bibinfo{year}{2019}), \bibinfo{pages}{386--402}.
\newblock


\bibitem[Han and Thorup(2002)]%
        {han2002sorting}
\bibfield{author}{\bibinfo{person}{Yijie Han} {and} \bibinfo{person}{Mikkel Thorup}.} \bibinfo{year}{2002}\natexlab{}.
\newblock \showarticletitle{Sorting integers in $O (n \sqrt{\log \log n})$ expected time and linear space}. In \bibinfo{booktitle}{\emph{IEEE Symposium on Foundations of Computer Science (FOCS’02)}}.
\newblock


\bibitem[Haussler and Welzl(1986)]%
        {haussler1986epsilon}
\bibfield{author}{\bibinfo{person}{David Haussler} {and} \bibinfo{person}{Emo Welzl}.} \bibinfo{year}{1986}\natexlab{}.
\newblock \showarticletitle{Epsilon-nets and simplex range queries}. In \bibinfo{booktitle}{\emph{Proceedings of the second annual symposium on Computational geometry}}. \bibinfo{pages}{61--71}.
\newblock


\bibitem[Haveliwala et~al\mbox{.}(1999)]%
        {haveliwala1999efficient}
\bibfield{author}{\bibinfo{person}{Taher Haveliwala} {et~al\mbox{.}}} \bibinfo{year}{1999}\natexlab{}.
\newblock \bibinfo{booktitle}{\emph{Efficient computation of PageRank}}.
\newblock \bibinfo{type}{{T}echnical {R}eport}. \bibinfo{institution}{Citeseer}.
\newblock


\bibitem[Juba et~al\mbox{.}(2015)]%
        {juba2015principled}
\bibfield{author}{\bibinfo{person}{Brendan Juba}, \bibinfo{person}{Christopher Musco}, \bibinfo{person}{Fan Long}, \bibinfo{person}{Stelios Sidiroglou-Douskos}, {and} \bibinfo{person}{Martin~C Rinard}.} \bibinfo{year}{2015}\natexlab{}.
\newblock \showarticletitle{Principled Sampling for Anomaly Detection.}. In \bibinfo{booktitle}{\emph{NDSS}}.
\newblock


\bibitem[Kent(2016)]%
        {kent2016cyber}
\bibfield{author}{\bibinfo{person}{Alexander~D Kent}.} \bibinfo{year}{2016}\natexlab{}.
\newblock \showarticletitle{Cyber security data sources for dynamic network research}.
\newblock In \bibinfo{booktitle}{\emph{Dynamic Networks and Cyber-Security}}. \bibinfo{publisher}{World Scientific}, \bibinfo{pages}{37--65}.
\newblock


\bibitem[Kumar and Minz(2014)]%
        {kumar2014feature}
\bibfield{author}{\bibinfo{person}{Vipin Kumar} {and} \bibinfo{person}{Sonajharia Minz}.} \bibinfo{year}{2014}\natexlab{}.
\newblock \showarticletitle{Feature selection}.
\newblock \bibinfo{journal}{\emph{SmartCR}} \bibinfo{volume}{4}, \bibinfo{number}{3} (\bibinfo{year}{2014}), \bibinfo{pages}{211--229}.
\newblock


\bibitem[Liberty et~al\mbox{.}(2016)]%
        {liberty2016stratified}
\bibfield{author}{\bibinfo{person}{Edo Liberty}, \bibinfo{person}{Kevin Lang}, {and} \bibinfo{person}{Konstantin Shmakov}.} \bibinfo{year}{2016}\natexlab{}.
\newblock \showarticletitle{Stratified sampling meets machine learning}. In \bibinfo{booktitle}{\emph{International conference on machine learning}}. PMLR, \bibinfo{pages}{2320--2329}.
\newblock


\bibitem[Mai et~al\mbox{.}(2006)]%
        {mai2006sampled}
\bibfield{author}{\bibinfo{person}{Jianning Mai}, \bibinfo{person}{Chen-Nee Chuah}, \bibinfo{person}{Ashwin Sridharan}, \bibinfo{person}{Tao Ye}, {and} \bibinfo{person}{Hui Zang}.} \bibinfo{year}{2006}\natexlab{}.
\newblock \showarticletitle{Is sampled data sufficient for anomaly detection?}. In \bibinfo{booktitle}{\emph{Proceedings of the 6th ACM SIGCOMM conference on Internet measurement}}. \bibinfo{pages}{165--176}.
\newblock


\bibitem[Rusu and Dobra(2009)]%
        {rusu2009sketching}
\bibfield{author}{\bibinfo{person}{Florin Rusu} {and} \bibinfo{person}{Alin Dobra}.} \bibinfo{year}{2009}\natexlab{}.
\newblock \showarticletitle{Sketching sampled data streams}. In \bibinfo{booktitle}{\emph{2009 IEEE 25th International Conference on Data Engineering}}. IEEE, \bibinfo{pages}{381--392}.
\newblock


\bibitem[Spielman and Teng(2004)]%
        {spielman2004Nibble}
\bibfield{author}{\bibinfo{person}{Daniel~A Spielman} {and} \bibinfo{person}{Shang-Hua Teng}.} \bibinfo{year}{2004}\natexlab{}.
\newblock \showarticletitle{Nearly-linear time algorithms for graph partitioning, graph sparsification, and solving linear systems}. In \bibinfo{booktitle}{\emph{Proceedings of the thirty-sixth annual ACM symposium on Theory of computing}}. \bibinfo{pages}{81--90}.
\newblock


\bibitem[Vu et~al\mbox{.}(2019)]%
        {vu2019feature}
\bibfield{author}{\bibinfo{person}{Loan~Thi Vu}, \bibinfo{person}{Lien~Thi Vu}, \bibinfo{person}{Nga~Thu Nguyen}, \bibinfo{person}{Phuong Thi~Thuy Do}, {and} \bibinfo{person}{Daniel Dao}.} \bibinfo{year}{2019}\natexlab{}.
\newblock \showarticletitle{Feature selection methods and sampling techniques to financial distress prediction for Vietnamese listed companies}.
\newblock \bibinfo{journal}{\emph{Investment Management and Financial Innovations}} \bibinfo{volume}{16}, \bibinfo{number}{1} (\bibinfo{year}{2019}), \bibinfo{pages}{276--290}.
\newblock


\bibitem[Wang et~al\mbox{.}(2021)]%
        {wang2021approximate}
\bibfield{author}{\bibinfo{person}{Hanzhi Wang}, \bibinfo{person}{Mingguo He}, \bibinfo{person}{Zhewei Wei}, \bibinfo{person}{Sibo Wang}, \bibinfo{person}{Ye Yuan}, \bibinfo{person}{Xiaoyong Du}, {and} \bibinfo{person}{Ji-Rong Wen}.} \bibinfo{year}{2021}\natexlab{}.
\newblock \showarticletitle{Approximate graph propagation}. In \bibinfo{booktitle}{\emph{Proceedings of the 27th ACM SIGKDD Conference on Knowledge Discovery \& Data Mining}}. \bibinfo{pages}{1686--1696}.
\newblock


\bibitem[Xing and Ghorbani(2004)]%
        {xing2004weighted}
\bibfield{author}{\bibinfo{person}{Wenpu Xing} {and} \bibinfo{person}{Ali Ghorbani}.} \bibinfo{year}{2004}\natexlab{}.
\newblock \showarticletitle{Weighted pagerank algorithm}. In \bibinfo{booktitle}{\emph{Proceedings. Second Annual Conference on Communication Networks and Services Research, 2004.}} IEEE, \bibinfo{pages}{305--314}.
\newblock


\bibitem[Yang et~al\mbox{.}(2019)]%
        {yang2019TEA}
\bibfield{author}{\bibinfo{person}{Renchi Yang}, \bibinfo{person}{Xiaokui Xiao}, \bibinfo{person}{Zhewei Wei}, \bibinfo{person}{Sourav~S Bhowmick}, \bibinfo{person}{Jun Zhao}, {and} \bibinfo{person}{Rong-Hua Li}.} \bibinfo{year}{2019}\natexlab{}.
\newblock \showarticletitle{Efficient estimation of heat kernel pagerank for local clustering}. In \bibinfo{booktitle}{\emph{Proceedings of the 2019 International Conference on Management of Data}}. \bibinfo{pages}{1339--1356}.
\newblock


\bibitem[Yi et~al\mbox{.}(2023)]%
        {yi2023optimal}
\bibfield{author}{\bibinfo{person}{Lu Yi}, \bibinfo{person}{Hanzhi Wang}, {and} \bibinfo{person}{Zhewei Wei}.} \bibinfo{year}{2023}\natexlab{}.
\newblock \showarticletitle{Optimal Dynamic Subset Sampling: Theory and Applications}. In \bibinfo{booktitle}{\emph{Proceedings of the 29th ACM SIGKDD Conference on Knowledge Discovery and Data Mining}}. \bibinfo{pages}{3116--3127}.
\newblock


\bibitem[Yin et~al\mbox{.}(2017)]%
        {yin2017MAPPR}
\bibfield{author}{\bibinfo{person}{Hao Yin}, \bibinfo{person}{Austin~R Benson}, \bibinfo{person}{Jure Leskovec}, {and} \bibinfo{person}{David~F Gleich}.} \bibinfo{year}{2017}\natexlab{}.
\newblock \showarticletitle{Local higher-order graph clustering}. In \bibinfo{booktitle}{\emph{Proceedings of the 23rd ACM SIGKDD international conference on knowledge discovery and data mining}}. \bibinfo{pages}{555--564}.
\newblock


\bibitem[Zenati et~al\mbox{.}(2018)]%
        {zenati2018adversarially}
\bibfield{author}{\bibinfo{person}{Houssam Zenati}, \bibinfo{person}{Manon Romain}, \bibinfo{person}{Chuan-Sheng Foo}, \bibinfo{person}{Bruno Lecouat}, {and} \bibinfo{person}{Vijay Chandrasekhar}.} \bibinfo{year}{2018}\natexlab{}.
\newblock \showarticletitle{Adversarially learned anomaly detection}. In \bibinfo{booktitle}{\emph{2018 IEEE International conference on data mining (ICDM)}}. IEEE, \bibinfo{pages}{727--736}.
\newblock


\bibitem[Zhang et~al\mbox{.}(2024)]%
        {zhang2024influence}
\bibfield{author}{\bibinfo{person}{Lingling Zhang}, \bibinfo{person}{Hong Jiang}, \bibinfo{person}{Ye Yuan}, {and} \bibinfo{person}{Guoren Wang}.} \bibinfo{year}{2024}\natexlab{}.
\newblock \showarticletitle{Influence Maximization in Hypergraphs by Stratified Sampling for Efficient Generation of Reverse Reachable Sets}.
\newblock \bibinfo{journal}{\emph{arXiv preprint arXiv:2406.01911}} (\bibinfo{year}{2024}).
\newblock


\end{thebibliography}
\end{document}